\documentclass[conference]{IEEEtran} 
\usepackage{stfloats}
\usepackage{graphicx}
\usepackage[most]{tcolorbox}
\tcbuselibrary{skins, breakable}
\usepackage{circledsteps}
\usepackage[probability, sets, operators, primitives]{cryptocode}
\usepackage{soul}
\usepackage{annotate-equations}
\usepackage{multicol}
\usepackage{multirow}
\usepackage{cancel}
\usepackage{pifont}
\usepackage{amsthm}
\usepackage{tabularx}
\newtheorem{theorem}{Theorem}
\usepackage{colortbl}

\usepackage{booktabs}
\usepackage[table]{xcolor}
\usepackage{array}
\usepackage{float}
\usepackage{subcaption}
\usepackage{hyperref}
\usetikzlibrary{shapes.geometric, arrows.meta, positioning, calc, fit, decorations.pathreplacing}

\newcommand{\AtkZero}{Attack~0}
\newcommand{\AtkOne}{Attack~1}
\newcommand{\AtkTwo}{Attack~2}
\newcommand{\AtkThree}{Attack~3}
\newcommand{\AtkZeroTitle}{Pre-MPC Information Leakage}
\newcommand{\AtkOneTitle}{Opening Fairness Violation}
\newcommand{\AtkTwoTitle}{On-Chain/Off-Chain Correlation}
\newcommand{\AtkThreeTitle}{Invalid-Input Griefing}

\begin{document}

\title{Cryptographic Security Is Not Enough: Privacy Gaps in the Renegade Decentralized Dark Pool}

\author{
\IEEEauthorblockN{
Prerna Arote\IEEEauthorrefmark{1},
Adrian Saiz\IEEEauthorrefmark{2},
Oriol Saguillo\IEEEauthorrefmark{1},
Lucianna Kiffer\IEEEauthorrefmark{1}
}

\IEEEauthorblockA{\IEEEauthorrefmark{1}IMDEA Networks Institute}

\IEEEauthorblockA{\IEEEauthorrefmark{2}EPFL}

\IEEEauthorblockA{
\IEEEauthorrefmark{1}\{prerna.arote, oriol.saguillo, lucianna.kiffer\}@networks.imdea.org\\
\IEEEauthorrefmark{2}adrian.saizdepedro@epfl.ch
}
}
\maketitle

\begin{abstract} 
Dark pools are designed to provide pre-trade privacy, liveness, and post-trade confidentiality --- concealing order flow before execution and limiting information leakage after. Decentralized dark pools, such as Renegade, aim to replicate these properties without custodial risk, using secure multi-party computation (MPC) and zero-knowledge proofs for private order matching and verifiable settlement.

We show that Renegade's cryptographic guarantees do not deliver these dark pool properties in practice. MPC-with-abort ensures correctness but not fairness: a party may learn the match result and abort without penalty, breaking pre-trade privacy. We demonstrate that the protocol's discovery layer further leaks trading intent before MPC even begins, and that sustained probing via selective abort can probabilistically reconstruct counterparty order history, threatening post-trade confidentiality. We also show that the absence of input-consistency checks prior to MPC execution enables a griefing attack using invalid state commitments requiring no real token holdings that continuously locks honest users' wallets and wastes compute, breaking liveness under sustained conditions.

We further analyze over 700,000 Renegade transactions on Base and probe the P2P layer, finding that the network is effectively centralized: 88\% of traffic routes through a handful of relayers, with only four nodes sustaining the P2P layer. Since relayers hold their users' wallet state in plaintext, this concentration means the system operates as a centralized orderbook in practice --- reproducing off-chain the information asymmetry that dark pools are designed to eliminate.

Together, our results show that cryptographic privacy does not imply dark pool security: pre-trade privacy, liveness, and post-trade confidentiality each require additional protocol-level guarantees beyond MPC correctness.

\end{abstract}

\section{Introduction}
\label{sec:intro}

Cryptocurrency markets process billions of dollars in daily trading volume, yet their underlying trading infrastructure remains fundamentally limited. Historically, most liquidity has been concentrated in centralized exchanges (CEXes)~\cite{coingecko2023report}, which offer efficient execution and confidentiality from other traders but require users to relinquish custody of their assets. This custodial model has repeatedly led to failures, including high-profile exchange collapses~\cite{gemini_mtgox, ftx_ray_declaration} and instances of internal price manipulation~\cite{gandal2018price}.

\textbf{Decentralized exchanges} (DEXes) aim to eliminate custodial risk by enabling trustless, non-custodial trading via smart contracts~\cite{hagele2024centralized}. However, removing intermediaries does not eliminate risk. The transparency of blockchain systems introduces new vulnerabilities: pending transactions and orders are publicly visible prior to execution, exposing traders to adversarial behaviors such as \textit{miner extractable value (MEV)}, where block producers reorder or front-run transactions for profit, as well as pre-trade information leakage and post-trade traceability~\cite{daian2020flash, heimbach2023ethereum, yang2025decentralization}.

These vulnerabilities share a common root cause: \textbf{public visibility of trading activity}. An alternative direction seeks to eliminate information leakage altogether.

\smallskip
\noindent\textbf{Decentralized dark pools.} In traditional financial markets, \textbf{dark pools} enable large trades to be executed without revealing order flow, reducing price impact and strategic exploitation. In contrast, decentralized finance exposes trading activity in real time, enabling adversaries to react to observed order flow. This gap has led to the emergence of \textbf{blockchain-based dark pools}, which aim to combine the non-custodial guarantees of DEXes with the privacy properties of traditional dark pools. By concealing order information during execution, these systems aim to eliminate the ``transparency tax'' and protect traders from adversarial strategies.

Dark pools differ in their underlying privacy architectures, each reflecting a distinct trade-off. Trusted execution environment (TEE)-based approaches (e.g., Tristero, Silhouette) achieve low latency but rely on hardware trust and remain vulnerable to side-channel attacks~\cite{tristero}\cite{silhouette}. Zero-knowledge proof (ZKP)-based systems (e.g., Penumbra, Panther) provide strong privacy guarantees but require external coordination for matching~\cite{penumbra,pantherprotocol}. Fully homomorphic encryption (FHE)-based designs (e.g., Singularity) offer strong theoretical privacy but remain impractical due to high computational overhead \cite{singularity}. 
Multi-party computation (MPC)-based approaches~\cite{da2022all, cartlidge2019mpc, da2022kicking, baum2021p2dex, govindarajan2022privacy, renegade} offer a different point in this design space, enabling private order matching without trusted hardware. Existing MPC-based systems, however, differ in their trust assumptions and execution architectures.
Other privacy-preserving DEXs use MPC for different purposes; for example, COMMON~\cite{garreta2023common} uses MPC for distributed decryption rather than matching and is outside our scope.

\textbf{Renegade} is one of the most complete public design of an end-to-end MPC-based decentralized dark pool, and the most mature one to have been deployed with substantial publicly observable trading history. Its protocol is documented in a public specification~\cite{renegade} and backed by an open-source implementation of the full stack, including the MPC matching engine, the proof system, and the settlement contracts. Between September 2024 and April 2026, its deployment on Arbitrum and Base processed over \$340M in cumulative trading volume~\cite{defillama_renegade}.

In May 2026, a vulnerability in its smart contract was exploited by white-hat attackers, resulting in the withdrawal of funds from affected accounts; this incident is unrelated to the privacy vulnerabilities we study \cite{cointelegraph_renegade_exploit,tangem_renegade_exploit}.
Renegade combines secure multi-party computation (MPC) and zero-knowledge succinct non-interactive arguments of knowledge (zkSNARKs) to enable private, non-custodial trade execution~\cite{renegade}: orders are matched via MPC over secret-shared inputs, and trades are settled on-chain using collaboratively generated proofs. By matching orders off-chain and settling on-chain, Renegade aims to mitigate MEV and front-running while providing end-to-end privacy across both pre-trade and post-trade phases. 


\smallskip
\noindent\textbf{The gap.}
Renegade is, to our knowledge, the only end-to-end MPC-based decentralized dark pool deployed in production.
Section~\ref{sec:mpc_dark_pools} compares Renegade with prior MPC-based systems and distinguishes vulnerabilities inherent to permissionless end-to-end architectures from those arising from Renegade's specific design choices.
However, cryptographic correctness alone does not imply that a decentralized dark pool achieves its intended privacy guarantees in practice.
While MPC ensures that matching is computed correctly if the protocol completes, it does not by itself address information leakage before execution, strategic behavior during protocol execution, or protocol availability under adversarial participation.
The current protocol fails each of these at distinct points: the handshake layer leaks order intent and wallet identifiers before MPC begins; the output reconstruction phase allows one party to receive the match result and abort before the other does; and the absence of pre-MPC input consistency checks allows an adversary to force repeated aborts at negligible cost, preventing honest participants from completing matches.

\medskip
\noindent\textbf{Our work.}
We formally analyze Renegade and show that its cryptographic correctness does not prevent an adversary with standard protocol capabilities from compromising privacy, fairness, and liveness.
We identify four attacks spanning the protocol timeline. \emph{Handshake probing} reveals active token pairs and wallet identifiers before MPC begins, violating pre-trade privacy.
\emph{Selective abort} at output reconstruction allows an adversary to learn the match result before the honest party and abort selectively, providing a zero-cost option on trades and enabling complete orderbook extraction through repeated queries.
\emph{On-chain/off-chain correlation} links the wallet identifiers exposed during probing to funding addresses, yielding fully deanonymized trading profiles.
Finally, \emph{invalid-input griefing} allows an adversary to complete a valid handshake but supply inconsistent MPC inputs, forcing deterministic aborts that lock honest wallets and waste non-reusable preprocessing at a fraction of the honest party's cost and with negligible locked capital.

\medskip
\noindent\textbf{Empirical and benchmark evidence.}
To bound the real-world attack surface, we analyze over 700,000 transactions on the Renegade Base deployment and probe the P2P layer. At the time of our measurements, the mainnet consisted of just four nodes in a single AWS availability zone, meaning attacks targeting the entire network require only four concurrent sessions. On-chain data further shows that relay-to-user mappings are visible in the public transaction record, making individual users straightforward to identify and target. Benchmarks confirm the attacks are practical: a full network orderbook scan via handshake probing completes in minutes, the selective abort attack reveals complete match economics in under 2 seconds, and the griefing attack achieves a $487\times$ time asymmetry, destroying 20\,GB of non-reusable preprocessing per 1.6-second attack cycle.

Together, these results show that cryptographic correctness is not sufficient for dark pool privacy guarantees.

\smallskip
\noindent\textbf{Our contributions.}
\begin{itemize}

\item \textbf{Security model.} We formalize the security properties a decentralized dark pool must provide, distinguishing core privacy properties (pre-trade privacy, output fairness, post-trade privacy) from operational requirements (input binding, liveness).

\item \textbf{Attacks.} We prove four attacks on Renegade, each violating one or more of these properties, spanning the full protocol timeline from counterparty discovery through settlement.

\item \textbf{Empirical analysis.} We analyze over 700,000 on-chain transactions and probe the P2P network, characterizing the real-world attack surface and network topology.

\item \textbf{Benchmarks.} We benchmark handshake proof generation, MPC session cost, and ZKP circuit evaluation, establishing the cost parameters and asymmetries governing each attack.

\item \textbf{Mitigations.} We discuss protocol-level defenses for each attack class and their tradeoffs.

\end{itemize}

\noindent\textbf{Organization.}
Section~\ref{sec:background} presents background on Renegade. Section~\ref{sec:properties} formalizes the security model and properties. Section~\ref{sec:attacks} presents the four attacks with formal proofs. Section~\ref{sec:empirical} presents the empirical analysis. Section~\ref{sec:benchmarks} benchmarks each attack component. Section~\ref{sec:mpc_dark_pools} compares Renegade with prior MPC-based dark pools. Section~\ref{sec:mitigations} discusses mitigations, and we conclude with related work in Section~\ref{sec:related_work}.

\section{Background}
\label{sec:background}
We begin by describing the Renegade protocol in detail. 

\subsection{Renegade Protocol Overview}  %

Renegade~\footnote{The original Renegade whitepaper is no longer hosted on the project's website; we cite the Wayback Machine snapshot from 12 April 2026: \url{https://web.archive.org/web/20260412233320/https://whitepaper.renegade.fi/}.}is a decentralized dark pool protocol designed for \textbf{privacy-preserving trading} 
of large-scale order execution. Unlike traditional decentralized exchanges, where orders and balances are publicly visible on-chain, Renegade aims to conceal order flow by separating trading intent from execution and settlement, combining off-chain secure computation with on-chain cryptographic verification.

The protocol consists of three main components---users, relayers, and a smart contract---detailed in Section~\ref{subsec:sysmodel_entities}. \emph{Users} hold private wallets encoding balances and orders; \emph{relayers} handle counterparty discovery, matching, and proof generation, with plaintext access to the wallets they manage; and the \emph{smart contract} is a global verifier storing only cryptographic commitments to wallet state.

Trading proceeds as follows. Users update their \textbf{wallet states on-chain} by submitting cryptographic commitments and zero-knowledge proofs that validate deposits, withdrawals, and order updates. Relayers then perform \textbf{off-chain matching} via secure MPC to identify compatible counterparties, with the goal of keeping order details such as price and quantity private during execution. Once a match is found, relayers jointly generate a zero-knowledge proof attesting to correct execution and submit it to the smart contract, which verifies correctness and updates the global state. To coordinate matching, relayers communicate over a peer-to-peer network using gossip-based message propagation and direct request-response channels.

In summary, Renegade's design goal is private trade execution via commitment-based state representation, MPC-based off-chain matching, and on-chain cryptographic verification. 
What defines private trade execution and whether Renegade lives up to this goal in practice is the subject of this paper.

\subsection{System Model and Entities}
\label{subsec:sysmodel_entities}

The Renegade protocol 
consists of three main components: users, relayers (clusters), and the smart contract.

\textbf{Users} are traders who wish to exchange one token for another at a specified price. To do so, a user places \emph{limit orders} specifying a token pair, direction (buy or sell), a price limit, and a quantity. The user's goal is to find a counterparty with a compatible order and settle the trade atomically, without revealing the order details prior to execution. Users maintain a \emph{private wallet} encoding their current token balances, outstanding orders, and authorization data. This state is not stored in plaintext on-chain; instead, it is represented via cryptographic commitments to preserve privacy. Users may either operate their own relayer or delegate wallet management to an external relayer.

\textbf{Relayers} are specialized nodes to which users delegate wallet management and order execution, since end-users are not expected to remain online to perform complex computations. Relayers have plaintext access to the wallets they manage, enabling them to observe balances and execute orders, making them trusted intermediaries. 
To ensure availability and scalability, relayers are organized into \textbf{relayer clusters}: replicated, fault-tolerant groups of nodes managing the same set of wallets. 

Matching occurs in one of two modes. When two compatible orders reside within the same cluster, they can be matched \emph{internally}, without MPC, since the cluster already holds both wallets in plaintext. When compatible orders belong to different clusters, relayers perform \emph{cross-cluster} matching via a maliciously secure two-party MPC protocol, keeping each cluster's order details hidden from the other. In both cases, relayers jointly generate zero-knowledge proofs required for atomic settlement on-chain. A simplified architecture is shown in Figure~\ref{fig:reneagde-architecture}. To mitigate the associated privacy and centralization risks, the protocol allows users to operate independent relayers, though in practice few users do so (see Section~\ref{sec:relay_behavior}).

\begin{figure}[h]
    \centering
    \resizebox{0.8\linewidth}{!}{\input{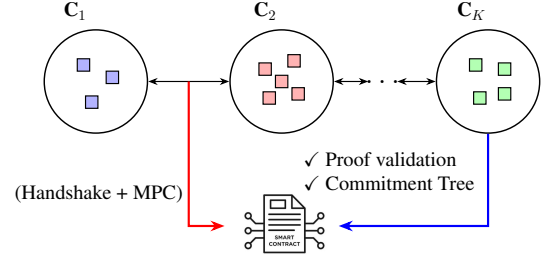}}
    \caption{Renegade cluster architecture. Relayer clusters hold user wallets and interact with the on-chain smart contract for settlement. Compatible orders across clusters are matched via MPC (\textcolor{red}{$\rightarrow$}); orders within the same cluster are matched internally and the state change committed on-chain (\textcolor{blue}{$\rightarrow$}).}
    \label{fig:reneagde-architecture}
\end{figure}

\textbf{Communication.}\label{sec:network}
Relayers communicate over a peer-to-peer network using three mechanisms: \textit{peer discovery}, allowing relayers to locate and connect to each other; \textit{counterparty discovery} via a gossip layer, over which all relayers broadcast order intent and compatible peers respond via a direct channel; and \textit{direct point-to-point communication} for handshake payloads and MPC execution.\footnote{Concretely, the network is built on libp2p~\cite{libp2p}, a modular networking framework providing transport, authentication, and routing primitives. Peer discovery uses Kademlia~\cite{maymounkov2002kademlia}, a structured distributed hash table that allows relayers to locate each other and maintain robust network connectivity. Order gossip uses GossipSub~\cite{gossipsub}, a publish-subscribe protocol that efficiently propagates messages to all participating nodes. Direct communication uses a RequestResponse protocol for authenticated point-to-point exchanges.}

\textbf{Smart contract.} The smart contract serves as the global state arbiter. To preserve privacy, it does not store account balances directly; 
instead, it maintains cryptographic commitments to wallet state, described in detail in Section~\ref{sec:trade_life_cycle}.

\begin{figure*}[!t]
    \centering
    \begin{subfigure}[t]{0.45\textwidth}
        \centering
        \resizebox{\linewidth}{!}{\input{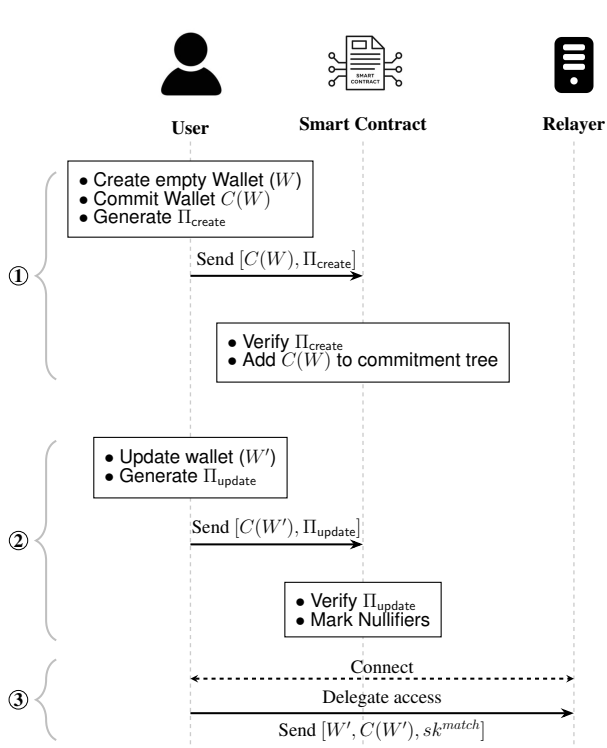}}
        \caption{Wallet creation through relayer delegation (steps 1--3).}
        \label{fig:trade_cycle_part1}
    \end{subfigure}
    \hfill
    \begin{subfigure}[t]{0.45\textwidth}
        \centering
        \resizebox{\linewidth}{!}{\input{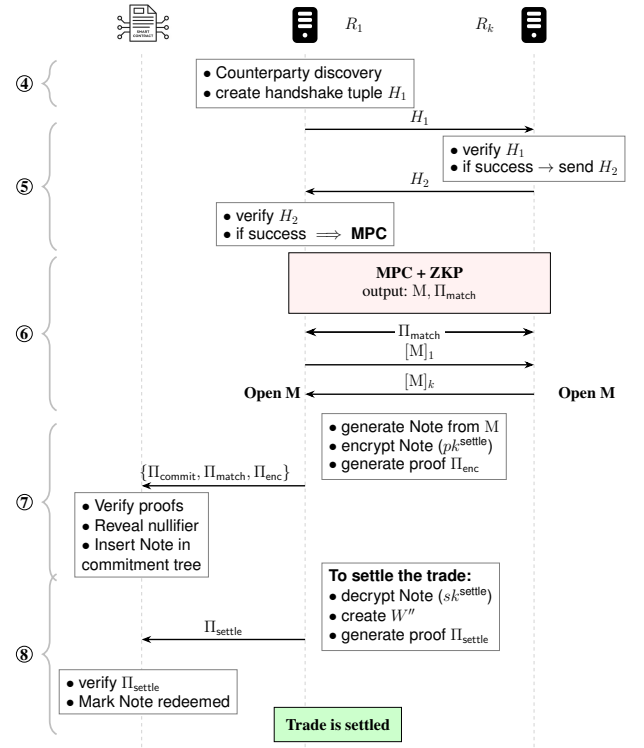}}
        \caption{Counterparty discovery through settlement (steps 4--8).}
        \label{fig:tradecycle_part2}
    \end{subfigure}
    \caption{The Renegade trade lifecycle.
    \textbf{(1)}~User creates a wallet and commits it on-chain.
    \textbf{(2)}~User updates wallet state (deposit or order) via commit-and-prove, revealing old nullifiers.
    \textbf{(3)}~User delegates the match key to a relayer, which manages the wallet in plaintext.
    \textbf{(4)}~Relayer broadcasts order intent over gossip and identifies a compatible counterparty.
    \textbf{(5)}~Relayers exchange handshake payloads, each committing to their order and wallet state via ZK proof.
    \textbf{(6)}~Relayers run SPDZ-style 2PC to evaluate the matching function and open the result~$M$.
    \textbf{(7)}~Relayers submit proofs on-chain; contract marks match nullifiers spent and records encrypted trade notes.
    \textbf{(8)}~Each relayer independently decrypts its note and submits an updated wallet state, consuming all remaining nullifiers.}
    \label{fig:trade_lifecycle}
\end{figure*}

\subsection{Trade Lifecycle}
\label{sec:trade_life_cycle}

\paragraph{State Representation and Commitments}

Renegade represents user state using commitments over structured wallet data~\cite{renegade}.

\paragraph{Wallet}
A wallet encodes balances, orders, fee approvals, and authorization keys. Formally, for some elliptic curve $\mathbb{G}$ over a prime field $\mathbb{F}$, a wallet $W$ is formally defined as a tuple containing the local state of a user:
\[
W = (B, O, F, K, r) \in \mathbb{F}^{2M_B} \times \mathbb{F}^{8M_O} \times \mathbb{F}^{5M_F} \times \mathbb{G}^{4} \times \mathbb{F},
\]
where $B$ records current token holdings; $O$ is the set of active limit orders the user wishes to execute, each specifying a token pair, direction, price, and quantity; $F$ contains fee authorizations granting the relayer permission to deduct a fee upon a successful match; $K$ is a hierarchy of public keys governing different levels of authority (described below); and $r$ is secret randomness used to derive commitments and nullifiers.

\paragraph{Key hierarchy}
Renegade uses a hierarchical key structure
\[
K = (\mathrm{pk}^{\mathrm{root}}, \mathrm{pk}^{\mathrm{match}}, \mathrm{pk}^{\mathrm{settle}}, \mathrm{pk}^{\mathrm{view}}),
\]
to restrict relayer capabilities. The root key authorizes deposits, withdrawals, and order updates, while the remaining keys enable matching, settlement, and state decryption. The root secret key is derived from the user’s Ethereum keypair, and remains known only to the user (Figure~\ref{fig:key_hierarchy}).

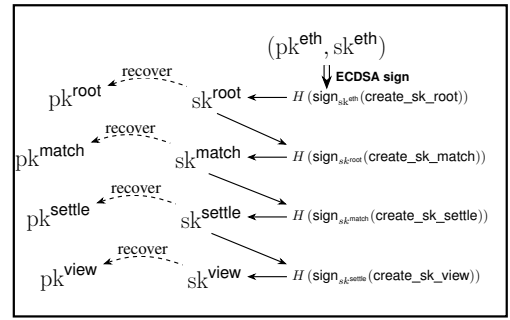
\begin{figure}[t] 
\centering                         \resizebox{0.75\linewidth}{!}{\begin{tikzpicture}[
    >=Stealth,
    thick,
    node distance = 0.8cm and 2cm, 
    keynode/.style={font=\LARGE\sffamily, anchor=center},
    labelnode/.style={font=\scriptsize\sffamily, fill=white, inner sep=1pt},
    double_arrow/.style={->, double, double distance=3pt, line width=0.8pt}
]

    \node[keynode] (eth_pair) at (4.5,-0.7) {$(\mathrm{pk}^{\text{eth}}, \mathrm{sk}^{\text{eth}})$};
    
    \node[font=\normalsize\sffamily, anchor=west] (act_root) at (3.5, -2) {
        $H\left( \text{sign}_{\mathrm{sk}^{\text{eth}}}(\text{create\_sk\_root}) \right)$
    };
    
    \draw[double_arrow] (eth_pair.south) -- (4.5, -1.8) 
        node[midway, right=0.1cm, font=\small\sffamily\bfseries] {ECDSA sign};

    \node[keynode, left=1cm of act_root] (sk_root) {$\mathrm{sk}^{\text{root}}$};
    \node[keynode, left=2cm of sk_root] (pk_root) {$\mathrm{pk}^{\text{root}}$};

    \node[font=\normalsize\sffamily, anchor=west] (act_match) at (3.5, -3.5) {
        $H\left( \text{sign}_{sk^{\text{root}}}(\text{create\_sk\_match}) \right)$
    };
    
    \node[keynode, left=1cm of act_match] (sk_match) {$\mathrm{sk}^{\text{match}}$};
    \node[keynode, left=2cm of sk_match] (pk_match) {$\mathrm{pk}^{\text{match}}$};

    \node[font=\normalsize\sffamily, anchor=west] (act_settle) at (3.5, -5) {
        $H\left( \text{sign}_{sk^{\text{match}}}(\text{create\_sk\_settle}) \right)$
    };
    
    \node[keynode, left=1cm of act_settle] (sk_settle) {$\mathrm{sk}^{\text{settle}}$};
    \node[keynode, left=2cm of sk_settle] (pk_settle) {$\mathrm{\mathrm{pk}}^{\text{settle}}$};

    \node[font=\normalsize\sffamily, anchor=west] (act_view) at (3.5, -6.5) {
        $H\left( \text{sign}_{sk^{\text{settle}}}(\text{create\_sk\_view}) \right)$
    };
    
    \node[keynode, left=1cm of act_view] (sk_view) {$\mathrm{sk}^{\text{view}}$};
    \node[keynode, left=2cm of sk_view] (pk_view) {$\mathrm{pk}^{\text{view}}$};

    \draw[<-, dashed] (pk_root) to [bend left=20] node[labelnode, above, font=\large] {recover} (sk_root);
    \draw[<-, dashed] (pk_match) to [bend left=20] node[labelnode, above, font=\large] {recover} (sk_match);
    \draw[<-, dashed] (pk_settle) to [bend left=20] node[labelnode, above, font=\large] {recover} (sk_settle);
    \draw[<-, dashed] (pk_view) to [bend left=20] node[labelnode, above, font=\large] {recover} (sk_view);

    \draw[->] (act_root) -- (sk_root);
    \draw[->] (act_match) -- (sk_match);
    \draw[->] (act_settle) -- (sk_settle);
    \draw[->] (act_view) -- (sk_view);

    \draw[->] (sk_root.south) -- (act_match.north west);
    \draw[->] (sk_match.south) -- (act_settle.north west);
    \draw[->] (sk_settle.south) -- (act_view.north west);

    \begin{scope}[on background layer]
        \node[draw, ultra thick, inner sep=15pt, fit=(eth_pair) (pk_view) (act_view)] (boundary) {};
    \end{scope}

\end{tikzpicture}} 
\caption{Key hierarchy in Renegade} \label{fig:key_hierarchy}
\end{figure} %

\begin{figure}[!t]
    \centering
    \resizebox{0.8\linewidth}{!}{\begin{tikzpicture}[
    member/.style={rectangle, draw, minimum height=1.2cm, minimum width=2.0cm, font=\Huge, line width=1pt},
    hash/.style={circle, draw, minimum size=0.9cm, line width=1pt},
    red_path/.style={draw=red!80!black, line width=1.5pt},
    blue_path/.style={draw=blue!80!black, line width=1.5pt},
    normal_line/.style={draw=black, line width=1pt},
    box_red/.style={rectangle, rounded corners=10pt, draw=black, fill=red!20, inner sep=10pt, align=center},
    box_blue/.style={rectangle, rounded corners=10pt, draw=black, fill=cyan!20, inner sep=10pt, align=center}
]

    \node[member] (C0)    at (0,0)     {$C_0$};
    \node[member, fill=red!20] (CW) at (3.0,0)   {$C(W)$};
    \node[member] (C2)    at (6.0,0)   {$C_2$};
    \node[member] (C3)    at (9.0,0)   {$C_3$};
    \node[member] (C4)    at (12.0,0)  {$C_4$};
    \node[member, fill=cyan!20] (CWp) at (15.0,0) {$C(W')$};
    \node[member] (null1) at (18.0,0)  {null};
    \node[member] (null2) at (21.0,0)  {null};

    \node[below=0.3cm of CW, align=center, font=\LARGE]  {Reveal\\Nullifiers};
    \node[below=0.3cm of CWp, align=center, font=\LARGE] {Commit\\to Wallet};

    \coordinate (H1_pos) at ($(C0)!0.5!(CW)$);
    \node[hash] (H1) at (H1_pos |- 0, 2.5) {};

    \coordinate (H2_pos) at ($(C2)!0.5!(C3)$);
    \node[hash] (H2) at (H2_pos |- 0, 2.5) {};

    \coordinate (H3_pos) at ($(C4)!0.5!(CWp)$);
    \node[hash] (H3) at (H3_pos |- 0, 2.5) {};

    \coordinate (H4_pos) at ($(null1)!0.5!(null2)$);
    \node[hash] (H4) at (H4_pos |- 0, 2.5) {};

    \coordinate (H12_pos) at ($(H1)!0.5!(H2)$);
    \node[hash] (H12) at (H12_pos |- 0, 5.0) {};

    \coordinate (H34_pos) at ($(H3)!0.5!(H4)$);
    \node[hash] (H34) at (H34_pos |- 0, 5.0) {};

    \coordinate (Root_pos) at ($(H12)!0.5!(H34)$);
    \node[hash, font=\Huge] (Root) at (Root_pos |- 0, 7.0) {$R_{\text{global}}$};

    \draw[normal_line] (H1) -- (C0);
    \draw[normal_line] (H2) -- (C2);
    \draw[normal_line] (H2) -- (C3);
    \draw[normal_line] (H3) -- (C4);
    \draw[normal_line] (H4) -- (null1);
    \draw[normal_line] (H4) -- (null2);
    \draw[normal_line] (H12) -- (H2);
    \draw[normal_line] (H34) -- (H4);

    \draw[red_path] (Root) -- (H12);
    \draw[red_path] (H12) -- (H1);
    \draw[red_path] (H1) -- (CW);

    \draw[blue_path] (Root) -- (H34);
    \draw[blue_path] (H34) -- (H3);
    \draw[blue_path] (H3) -- (CWp);

    \node[box_red, above left=0.6cm and -1.0cm of H12] (nullBox) {
        \Huge Nullifiers \\ \\
        \Huge $N^{\text{wallet-spend}}(W)$ \\ 
        \Huge $N^{\text{wallet-match}}(W)$ 
    };

    \node[box_blue, above right=2.1cm and -1.0cm of H34] (piBox) {
        \Huge ZKP of commitment $\pi$ 
    };

\end{tikzpicture}}
    \caption{Commitment-based state representation}
    \label{fig:commit-reveal}
\end{figure}
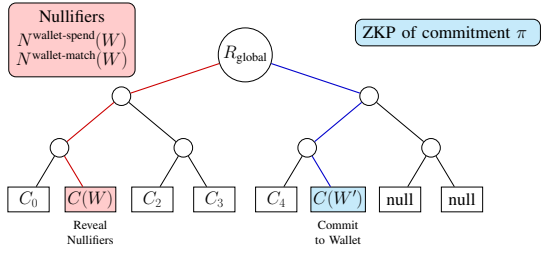

\paragraph{Commitment-based state}
Wallets are stored on-chain via commitments
\[
C(W) := H(\mathcal{H}(B) \| \mathcal{H}(O) \| \mathcal{H}(F) \| \mathcal{H}(K) \| r),
\]
where $H$ is a SNARK-friendly hash and $\mathcal{H}$ denotes Merkle hashing. Commitments are inserted into an append-only Merkle tree, and the contract maintains only the global root (Figure~\ref{fig:commit-reveal}).

\paragraph{State transitions}
Because wallet contents are hidden behind commitments, the contract cannot inspect or validate them directly. Updates from $W$ to $W’$ instead follow a commit-and-prove model: the user submits $C(W’)$ along with a zero-knowledge proof attesting to the validity of the transition and membership of $C(W)$ in the commitment tree.

\paragraph{Nullifiers and replay protection}
Because commitments are hiding, the contract cannot prevent reuse of the same wallet state. %
To address this, each wallet state $W$ derives two unique nullifiers from its commitment $C(W)$ and secret randomness $r$. These nullifiers are revealed and recorded upon consumption, allowing the contract to reject any subsequent attempt to reuse the same wallet state.
The spend nullifier $N^{\textsf{wallet-spend}}(W) := H(C(W) \| r)$ permanently invalidates $W$, used for deposits, withdrawals, and cancellations. The match nullifier $N^{\textsf{wallet-match}}(W) := H(C(W) \| r+1)$ temporarily locks it during an in-progress trade --- preventing the same funds from being pledged to two simultaneous matches, but leaving the state intact for settlement.

A trade proceeds in two phases: an \emph{off-chain matching phase}, in which relayers discover a compatible counterparty and run the MPC protocol, and an \emph{on-chain settlement phase}, in which the resulting proofs are submitted and wallet states are updated. Figure~\ref{fig:trade_cycle_part1} covers wallet setup and relayer delegation (steps 1--3); Figure~\ref{fig:tradecycle_part2} covers counterparty discovery through settlement (steps 4--8). 
Appendix~\ref{appendix:trade_lifecycle} gives a complete description; we focus here on steps 4--6 (counterparty discovery, handshake, and MPC execution), which underlie the attacks analyzed in Section~\ref{sec:attacks}.

\paragraph{Counterparty discovery (step 4)}
When a relayer has an active order, it gossips a bid identifier and its own peer ID to the network. Other relayers receive this message and, if they may have open orders, connect directly to the advertising relayer and initiate a lightweight handshake to learn the token pair of the bid. If the token pair is compatible, the relayers proceed to the full handshake (step 5); otherwise the connection is dropped.

\paragraph{Handshake (step 5)}
To initiate a match, relayer $\mathcal{R}_1$ sends a handshake payload
\[
H_1 = \big(\Pi_{\mathsf{commit}},\; N^{\textsf{wallet-match}}(W'),\; \mathrm{pk}^{\textsf{settle}},\; H_o, H_b, H_f, H_r\big),
\]
where $(H_o, H_b, H_f, H_r)$ are hiding commitments to the selected order, balance, fee, and randomness, and $\Pi_{\mathsf{commit}}$ proves they are well-formed and correspond to a valid, unspent wallet. The wallet-match nullifier $N^{\textsf{wallet-match}}(W')$ identifies the wallet state being committed to the match. The counterparty $\mathcal{R}_2$ verifies the proof and nullifier freshness before responding with its own payload $H_2$.

\paragraph{MPC execution and opening (step 6)}
After exchanging handshakes, $\mathcal{R}_1$ and $\mathcal{R}_2$ run a maliciously secure SPDZ-style two-party protocol over their secret-shared orders, balances, and fees. The circuit evaluates the matching function and produces a match result $M$ if the orders are compatible, along with a zero-knowledge proof $\Pi_{\mathsf{match}}$ binding $M$ to the public handshake commitments. The parties then perform an \emph{opening}: they exchange their shares of $M$, with SPDZ MAC verification ensuring integrity. Upon successful verification, both parties obtain $M$ in the clear.

\paragraph{On-chain settlement (steps 7--8)}
If $M$ indicates a valid match, the relayers encrypt their trade notes and jointly submit the proofs to the smart contract. The contract verifies correctness, marks the wallet-match nullifiers as spent, and inserts the encrypted notes into the commitment tree. Each relayer then independently settles by decrypting its note and submitting an updated wallet state on-chain, consuming the remaining nullifiers.

\section{Security Model and Properties}
\label{sec:properties}

Dark pools aim to enable trading without revealing sensitive order information prior to execution, while ensuring correct and fair matching, limiting selective advantage, providing timely execution, and remaining robust against malicious or economically rational adversaries. 
Although Renegade claims to provide these guarantees~\cite{renegade}, it does not define them explicitly. As the first end-to-end MPC-based dark pool in production, Renegade has no formal security model, and to our knowledge none exists for MPC-based dark pool matching more broadly.

We address this gap by formalizing the security properties under a two-party model: a single matching session between relayers $\mathcal{R}_1$ and $\mathcal{R}_2$. In Section~\ref{sec:attacks}, we discuss the impact of each attack depending on whether the victim relayer represents a single user or aggregates multiple users. 
We distinguish between \textbf{core privacy properties}, which capture the fundamental privacy and fairness guarantees expected of a dark pool, and \textbf{operational properties}, which capture mechanism-level requirements needed to realize these properties in adversarial environments.

\medskip
\noindent\textbf{Protocol Model.}
We model the Renegade matching protocol $\Pi$ as a two-party protocol between relayers $\mathcal{R}_1$ and $\mathcal{R}_2$. Each relayer $\mathcal{R}_i$ holds a private order
\[
  x_i = (\mathit{pair}_i,\; \mathit{dir}_i,\; \mathit{price}_i,\; \mathit{qty}_i),
\]
where $\mathit{pair}_i$ is the token pair, $\mathit{dir}_i \in \{\mathsf{buy}, \mathsf{sell}\}$ the trade direction, $\mathit{price}_i$ the limit price, and $\mathit{qty}_i$ the quantity. The matching functionality $\mathcal{F}$ outputs a match tuple
\[
  M = (\hat{m}_1, \hat{m}_2, \hat{v}_1, \hat{v}_2, \hat{d}, f_1, f_2)
\]
if the orders are compatible, and $\perp$ otherwise, where $\hat{m}_i$ are the matched asset amounts, $\hat{v}_i$ the corresponding values, $\hat{d}$ the execution direction, and $f_i$ the relayer fees.

The protocol proceeds in two phases. In the \emph{discovery phase}, relayers exchange handshake messages over a peer-to-peer gossip network to identify compatible counterparties. In the \emph{execution phase}, matched relayers run the MPC protocol, generate collaborative zero-knowledge proofs, and reconstruct the output $M$. The underlying MPC provides malicious security with abort: correctness holds conditioned on completion, but output delivery is not guaranteed --- a party may abort at any point without on-chain penalty. The adversary's view covers both phases.

\medskip
\noindent\textbf{Adversary Model.}
We consider an adversary $\mathcal{A}$ that corrupts one relayer (without loss of generality $\mathcal{R}_2$) and may deviate arbitrarily from the protocol. The adversary's capabilities are those of any standard protocol participant: it can establish connections, send and receive gossip and handshake messages, and participate in the MPC. No additional network access or side-channel information is assumed. In particular, the attacks presented in Section~\ref{sec:attacks} are \emph{information-theoretic}: they require no computational assumptions and succeed with advantage~$1$.

\medskip
\noindent\textbf{Core Privacy Properties:}

\begin{itemize}
\item \noindent\textbf{Pre-trade Privacy.}
$\Pi$ satisfies pre-trade privacy if an adversary cannot distinguish between two honest input vectors that produce the same output.

Formally, for any two input vectors $(x_1, x_2)$ and $(x_1', x_2)$ with $\mathcal{F}(x_1, x_2) = \mathcal{F}(x_1', x_2)$,
\[
\mathsf{view}_{\mathcal{A}}^\Pi(x_1, x_2) \;\approx_c\; \mathsf{view}_{\mathcal{A}}^\Pi(x_1', x_2),
\]
where the view includes all messages observed during both the discovery and execution phases, prior to output reconstruction.

\item \noindent\textbf{Output Fairness.}
$\Pi$ satisfies output fairness if whenever $\mathcal{A}$ obtains the match result $M$, the honest party $\mathcal{R}_1$ also receives $M$, except with negligible probability.

\emph{Note:} Output fairness is not achievable by MPC-with-abort alone~\cite{cleve1986limits}; it requires additional structure such as on-chain penalties or a trusted delivery mechanism.

\item \noindent\textbf{Post-trade Privacy.}
$\Pi$ satisfies post-trade privacy if, after execution, the adversary learns nothing beyond its own input $x_2$ and its output $y_2$ (which may be $M$ or $\perp$). Formally, there exists a PPT simulator $\mathcal{S}$ such that for all inputs $(x_1, x_2)$:
\[
\mathsf{view}_{\mathcal{A}}^\Pi(x_1, x_2) \;\approx_c\; \mathcal{S}(x_2,\, y_2).
\]

\end{itemize}

\noindent Together, these three properties cover the full protocol timeline: pre-trade privacy governs the period before the output is revealed, output fairness governs the opening itself, and post-trade privacy governs what is learnable thereafter.

\medskip
\noindent\textbf{Operational Properties.}
The following properties capture system-level requirements for correct and available protocol execution.

\begin{itemize}

\item \textbf{Correctness.}
$\Pi$ satisfies correctness if, whenever both parties complete the protocol, they output $\mathcal{F}(x_1, x_2)$, except with negligible probability.

\item \textbf{Input Binding.}
$\Pi$ satisfies input binding if each party’s input is fixed upon submission, and the protocol execution remains consistent with this input, except with negligible probability.

\item \noindent\textbf{Liveness.}
$\Pi$ satisfies $\Delta$-liveness if, for any adversary $\mathcal{A}$ and any two honest relayers $\mathcal{R}_i$, $\mathcal{R}_j$ holding compatible orders, at least one of $\mathcal{R}_i$ or $\mathcal{R}_j$ completes a match with some honest counterparty within bounded time $\Delta$. In particular, adversarial behavior cannot indefinitely prevent the protocol from making progress: whenever a valid match exists, at least one party to it will settle within bounded time.

\end{itemize}

We emphasize that malicious security of the underlying MPC guarantees correctness throughout: the protocol never produces an incorrect output. However, correctness alone is insufficient for dark pool security. In Section~\ref{sec:attacks}, we show that all remaining properties can be violated by an adversary with standard protocol capabilities.

\section{Attacks on Renegade}
\label{sec:attacks}
We present a series of attacks demonstrating that protocol $\Pi$ fails to satisfy the
security properties defined in Section~\ref{sec:properties}. The adversary model is as
defined there: $\mathcal{A}$ corrupts relayer $\mathcal{R}_2$ and has the capabilities
of any standard protocol participant. For clarity, we model each relayer as representing
a single user; where the multi-user case changes the attack impact we note it explicitly.

\subsection{\AtkZero: \AtkZeroTitle}
\label{Sec:attack_prempc_leak}

Before any MPC execution, $\mathcal{A}$ can determine which token pairs and directions $\mathcal{R}_1$ has open orders for, along with the wallet-state identifier of each responding wallet. This leakage occurs entirely at the handshake layer and violates pre-trade privacy.

\medskip
\noindent\textbf{Attack.}
Trading intent is revealed at the handshake layer: $\mathcal{R}_1$ responds to a handshake query only if it holds an order compatible with the queried $(p, d)$. $\mathcal{A}$ proceeds as follows:

\begin{enumerate}
    \item For each candidate token pair $p$ and direction $d \in \{\mathsf{buy}, \mathsf{sell}\}$, $\mathcal{A}$ sends a handshake query to $\mathcal{R}_1$ claiming to hold a matching-direction order for $(p, d)$.
    \item If $\mathcal{R}_1$ responds, $\mathcal{A}$ receives the handshake payload

   \[
    H_1=\bigl(
    \Pi_{\mathsf{commit}},
   N^{\textsf{wallet-match}}(W'), 
  \mathrm{pk}^{\textsf{settle}}_{\mathcal{R}_1},
  H_o,H_b,H_f,H_r
  \bigr)
   \]
    and learns: (a) that $\mathcal{R}_1$ holds an active order with $\mathit{pair}_1 = p$ and $\mathit{dir}_1 = \bar{d}$, and (b) the nullifier $N^{\textsf{wallet-match}}(W')$, a deterministic function of the wallet state $W'$ that serves as a stable pseudonymous identifier for the responding wallet.
    \item $\mathcal{A}$ repeats across all candidate $(p, d)$ pairs.
\end{enumerate}

The set of candidate pairs to probe is small in practice: on-chain deposit data bounds the universe of tokens active on the network (see Section~\ref{sec:empirical}).

\smallskip
\begin{theorem}
Protocol $\Pi$ does not satisfy pre-trade privacy.
\end{theorem}
\begin{proof}[Proof (sketch)] Let $x_1, x_1'$ differ only in token pair ($\mathit{pair} \neq \mathit{pair}'$), with $x_2$ chosen so $\mathcal{F}(x_1, x_2) = \mathcal{F}(x_1', x_2) =  \perp$. $\mathcal{A}$ issues a sell handshake for $\mathit{pair}$ to $\mathcal{R}_1$: it receives a response under $x_1$ but not under $x_1'$, distinguishing the two views with advantage $1$. \qedhere
\end{proof}

\smallskip
\noindent\textbf{Corollary (wallet-state leakage).} Each handshake response also reveals the wallet-match nullifier $N^{\textsf{wallet-match}}(W')$, a deterministic function of the wallet state that changes precisely when the wallet’s orders or balances are updated. Since all orders of a single user belong to the same wallet, all responding probes at a given point in time return the same nullifier. Consequently, the number of distinct nullifiers collected across all $(p,d)$ probes equals the number of distinct users at $\mathcal{R}_1$ with active orders. If $\mathcal{R}_1$ aggregates multiple users, $\mathcal{A}$ learns an exact count of active users and holds a stable pseudonymous identifier for each.

\medskip
\noindent\fcolorbox{black}{white!12}
{\parbox{\dimexpr\linewidth-2\fboxsep-2\fboxrule\relax}{%
\textbf{What \AtkZero{} reveals about $\mathcal{R}_1:$} By probing $\mathcal{R}_1$ across all candidate token pairs, $\mathcal{A}$ learns: (i) which pairs and directions $\mathcal{R}_1$ has active orders for; (ii) the number of users with open orders, given by the count of distinct wallet-match nullifiers returned.
}}

\medskip
\noindent\colorbox{gray!12}
{\parbox{\dimexpr\linewidth-2\fboxsep\relax}{%
\textbf{Network-wide corollary.} Since any peer is addressable, applying this procedure across all known relayers yields a real-time map of trading intent and active wallet states across the entire network, updated continuously without any MPC participation.
}}

\subsection{\AtkOne: \AtkOneTitle}
\label{Sec:attack_opening_fairness}
During the MPC output reconstruction phase, $\mathcal{A}$ withholds its share while receiving $\mathcal{R}_1$'s share of the matched tuple, enabling unilateral reconstruction of the match tuple. Consequently, $\mathcal{A}$ learns the trade price and size in addition to the trade intent, violating output fairness and post-trade privacy.

\medskip
\noindent\textbf{Attack.}
Following MPC execution and collaborative proof generation, $\mathcal{R}_1$ and $\mathcal{A}$ hold additive secret shares $[M]_1, [M]_2$ of the match tuple $[M]$ as:
\[
M = (\hat{m}_1, \hat{m}_2, \hat{v}_1, \hat{v}_2, \hat{d}, f_1, f_2).
\]
In order to reconstruct $M$ during opening phase, $\mathcal{R}_1$ and $\mathcal{A}$ must exchange $[M]_1$ and $[M]_2$ with each other.
$\mathcal{A}$ proceeds as follows:
\begin{enumerate}
    \item $\mathcal{A}$ waits for $\mathcal{R}_1$ to send its share $[M]_1$ as part of the output reconstruction phase.
    \item Upon receiving $[M]_1$, $\mathcal{A}$ reconstructs the full match $M$ using its local share $[M]_2$ and received share $[M]_1$.
    \item $\mathcal{A}$ aborts without sending $[M]_2$ to $\mathcal{R}_1$ if the trade is not favorable, otherwise proceeds to settle the trade.
\end{enumerate}

\medskip
\noindent\colorbox{gray!12}{\parbox{\dimexpr\linewidth-2\fboxsep\relax}{%
\textbf{Remark (free option).} Independent of the privacy violations below, this attack grants $\mathcal{A}$ a zero-premium option on every matched trade. By aborting unfavorable executions and settling favorable ones, $\mathcal{A}$ systematically extracts value from honest counterparties.
}}

\medskip
\begin{theorem}
Protocol $\Pi$ does not satisfy output fairness.
\end{theorem}

\begin{theorem}
\label{thm:post-trade-privacy}
Protocol $\Pi$ does not satisfy post-trade privacy.
\end{theorem}

\begin{proof}[Proof (sketch)]
Let $\mathcal{R}_1$ and $\mathcal{A}$ hold inputs $x_1, x_2$ with $\mathcal{F}(x_1, x_2)=M \neq \perp$. After execution both parties hold additive shares $([M]_1, [M]_2)$. During reconstruction, $\mathcal{R}_1$ sends $[M]_1$ to $\mathcal{A}$, which reconstructs $M$ from $([M]_1,[M]_2)$ and aborts without returning $[M]_2$, setting $y_2 = \perp$.

\smallskip\noindent\emph{Output fairness.} $\mathcal{A}$ holds $M$ while $\mathcal{R}_1$ cannot reconstruct it:
\[
\Pr[\mathcal{A} \text{ obtains } M] = 1; \quad \Pr[\mathcal{R}_1 \text{ obtains } M] = 0.
\]

\noindent\emph{Post-trade privacy.} $\mathcal{A}$’s view determines $M$ while its output is $\perp$. Any simulator $\mathcal{S}(x_2, \perp)$ cannot reconstruct $M$, so
\[
\mathsf{view}_{\mathcal{A}}^\Pi(x_1, x_2) \not\approx_c \mathcal{S}(x_2, \perp). \qedhere
\]
\end{proof}

\medskip
\noindent\textbf{Corollary (orderbook extraction).}
For an honest relayer $\mathcal{R}_1$ holding a private orderbook
\[
x_i = (\mathit{pair}_i,\; \mathit{dir}_i,\; \mathit{price}_i,\; \mathit{qty}_i),
\]
the adversary $\mathcal{A}$ can reconstruct the entire orderbook without committing to any trade.

\begin{proof}[Proof (sketch)]
From the handshake phase in \AtkZero, $\mathcal{A}$ obtains all $(\mathit{pair}_i,\mathit{dir}_i)$ for which $\mathcal{R}_1$ responds.
From \AtkOne, each matched MPC execution reveals $(\mathit{price}_i,\mathit{qty}_i)$ upon opening.
For each identified $(\mathit{pair}_i,\mathit{dir}_i)$, $\mathcal{A}$ fixes a matching direction and varies $(\mathit{price}_2,$ $\mathit{qty}_2)$ across repeated MPC executions. Each run reveals whether the hidden $(\mathit{price}_i,\mathit{qty}_i)$ is above or below the chosen values, yielding a comparison oracle.
Thus, binary search recovers $(\mathit{price}_i,\mathit{qty}_i)$ using logarithmic number of queries. Repeating over all $(\mathit{pair}_i,\mathit{dir}_i)$ yields the full orderbook. 
\end{proof}

\medskip
\noindent\fcolorbox{black}{white!12}
{\parbox{\dimexpr\linewidth-2\fboxsep-2\fboxrule\relax}{%
\textbf{What \AtkOne{} reveals about $\mathcal{R}_1:$}
$\mathcal{A}$ learns: (i) the exact price and quantity of the matched order; (ii) via repeated binary search, the complete orderbook of $\mathcal{R}_1$. 
}}

\subsection{\AtkTwo: \AtkTwoTitle}
\label{sec:on_offchain_correlationattack}
By combining active wallet probing with passive on-chain monitoring, $\mathcal{A}$ links anonymous wallet states to external funding addresses and reconstructs per-wallet balance evolution. This breaks post-trade privacy by enabling full wallet deanonymization.

\medskip
\noindent\textbf{Attack.}
The attack exploits the fact that wallet nullifiers change deterministically with deposits and withdrawals, while on-chain transfers reveal timing and counterparty addresses. $\mathcal{A}$ proceeds as follows:

\begin{enumerate}
  \item From \AtkZero{}, $\mathcal{A}$ obtains the set of active wallet nullifiers $N^{\textsf{wallet-match}}(W')$ and their grouping over relayers.
    
   \item $\mathcal{A}$ repeatedly probes all active wallets and records their current nullifiers.
    
    \item $\mathcal{A}$ passively observes ERC-20 transfer events involving the protocol contract, including $(\text{token}, \text{amount}, \text{direction}, \text{addr})$.
    
    \item After each observed transfer, $\mathcal{A}$ re-probes wallets and identifies the unique wallet whose nullifier changes.
    
    \item This establishes a linkage between the wallet nullifier and the on-chain address responsible for the transfer.
\end{enumerate}

\medskip
\noindent\textbf{Corollary (on-chain/off-chain linkage).}
The post-trade privacy violation of Theorem~\ref{thm:post-trade-privacy} enables a stronger consequence under passive on-chain monitoring. Since $\mathcal{A}$ holds the wallet nullifiers in plaintext, and each deposit or withdrawal produces a unique nullifier update alongside a visible ERC-20 transfer $(\mathsf{token}, \mathsf{amount},$ $\mathsf{direction}, \mathsf{addr})$, correlating the two yields a direct linkage
\[
N^{\textsf{wallet-match}}(W') \;\longleftrightarrow\; \mathsf{addr}.
\]
Repeated correlation across events resolves ambiguity and yields a complete mapping from wallet nullifiers to funding addresses and balance evolution.\medskip

\noindent\fcolorbox{black}{white!12}
{\parbox{\dimexpr\linewidth-2\fboxsep-2\fboxrule\relax}{%
\textbf{What \AtkTwo{} reveals about $\mathcal{R}_1$:} For each active wallet, $\mathcal{A}$ learns: (i) its funding address; (ii) its token balances over time; and (iii) its associated orders (from \AtkOne{}), yielding a fully deanonymized trading profile.
}}

\medskip
\noindent\colorbox{gray!12}
{\parbox{\dimexpr\linewidth-2\fboxsep\relax}{%
\textbf{Network-wide corollary.} Applying this attack across relayers yields a global mapping from anonymous wallets to funding addresses, along with their balances and trading activity, enabling network-wide deanonymization.
}}

\subsection{\AtkThree: \AtkThreeTitle}
$\mathcal{A}$ initiates MPC executions using inputs inconsistent with those committed during the handshake, forcing $\mathcal{R}_1$ to execute costly MPC computations that deterministically abort. Each such execution locks $\mathcal{R}_1$'s wallet and wastes compute on a match that can never succeed. Extending this to all relayers simultaneously breaks liveness: no compatible pair can make progress while all wallets remain locked.

\medskip
\noindent\textbf{Attack.}
The adversary $\mathcal{A}$ exploits the absence of pre-MPC consistency checks between the on-chain wallet commitment and the inputs supplied to the MPC. Specifically:
\begin{enumerate}
   
    \item $\mathcal{A}$ performs a valid handshake with $\mathcal{R}_1$ using a committed wallet state $W = (o, b, f, r)$ with a correct commitment, unspent state, and sufficient balance to cover fees.
    
    \item $\mathcal{A}$ submits MPC input $x_2' = (o^\ast, b, f, r)$ such that $o^\ast \neq o$ and $o^\ast \notin W$.
    
    \item $\mathcal{F}$ executes the MPC protocol on $(x_1, x_2')$; the inconsistency is detected only within the circuit, causing $\mathcal{F}$ to output $\perp$.
    
     \item Since each session locks $\mathcal{R}_1$'s wallet until abort and $\mathcal{A}$ incurs no penalty for repeating, $\mathcal{A}$ maintains a continuous stream of invalid sessions, keeping $\mathcal{R}_1$ locked indefinitely.
    
\end{enumerate}

\begin{theorem}
Protocol $\Pi$ does not satisfy input binding.
\end{theorem}

\begin{proof}[Proof (sketch)]
Let $\mathcal{A}$ complete a valid handshake with $\mathcal{R}_1$ on committed state $W$, but supply inconsistent MPC inputs $x_2' \neq x_2$. The handshake authenticates $W$, but the MPC proceeds on $x_2' \notin W$. Since no check enforces consistency between $W$ and MPC inputs prior to execution, the inconsistency is detected only inside the circuit, yielding $\mathcal{F}(x_1, x_2') = \perp$, violating input binding.
\end{proof}

\noindent\textbf{Corollary (liveness violation).}
Protocol $\Pi$ does not satisfy $\Delta$-liveness for any $\Delta$.
\begin{proof}[Proof (sketch)]
From the theorem above, each invalid session locks a relayer’s wallet until abort, and $\mathcal{A}$ incurs no penalty for repeating. Take any compatible pair $(\mathcal{R}_i, \mathcal{R}_j)$. $\mathcal{A}$ simultaneously maintains a continuous stream of invalid sessions against both $\mathcal{R}_i$ and $\mathcal{R}_j$, keeping both wallets locked indefinitely. Neither party can enter the matching phase, so neither completes a match within any $\Delta$, violating $\Delta$-liveness.
\end{proof}

\medskip
\noindent\fcolorbox{black}{white!12}
{\parbox{\dimexpr\linewidth-2\fboxsep-2\fboxrule\relax}{%
\textbf{Griefing impact of \AtkThree{} on $\mathcal{R}_1$:} $\mathcal{R}_1$'s wallet is kept locked and its compute wasted on executions that can never succeed. The attack is cheap to sustain: while $\mathcal{R}_1$ must hold and lock real token balances, $\mathcal{A}$ requires only minimal capital.\footnotemark{}
}}
\footnotetext{A valid handshake requires a committed wallet state with an unspent nullifier --- no real trading capital need be held. The honest party bears the full cost of capital lock-up.}

\medskip
\noindent\colorbox{gray!12}
{\parbox{\dimexpr\linewidth-2\fboxsep\relax}{%
\textbf{Network-wide corollary (DoS amplification).} By initiating concurrent invalid executions across relayers, $\mathcal{A}$ can occupy MPC execution capacity network-wide, preventing timely completion of valid matches and degrading overall system throughput.
}}

\medskip
The above attacks expose fundamental protocol-level vulnerabilities. We now examine their implications in the deployed Renegade network through empirical measurement.

\section{Empirical Analysis of the Renegade Network}
\label{sec:empirical}

A natural question is how the attacks presented in Section~\ref{sec:attacks} operate against the real Renegade
network. We begin by characterizing the P2P layer and find that the network is far from the decentralized architecture the
protocol assumes. With only a handful of active relayers, the system operates in practice
closer to a centralized orderbook than a distributed dark pool, concentrating the privacy risk
described in our attacks on very few nodes.
We turn to on-chain data to characterize the economic scale of the current deployment, establishing a baseline for the threat model as the protocol grows toward a genuinely decentralized setting.

\subsection{P2P Network Structure}

To characterize the network the attacks operate against, we deployed a relayer node. Since no public bootstrap nodes exist, we could not join the P2P network directly to enumerate peers. Instead, we used \textbf{ffuf}~\cite{ffuf} to fuzz the Renegade relayer REST API; most endpoints returned errors or required specialized headers, but we identified the public \texttt{/v0/network} endpoint\footnote{Base: \url{https://base-mainnet.relayer.renegade.fi/v0/network}; Arbitrum: \url{https://arbitrum-one.relayer.renegade.fi/v0/network}. Both endpoints return the same set of nodes, indicating a single shared P2P network across Renegade's chain deployments.}, which returns the complete global state of the gossip network --- all active peer IDs, cluster assignments, and addresses. We supplemented this with a crawl using Nebula~\cite{nebula}, a network crawler for libp2p-based peer-to-peer networks. We verified accuracy via a loopback test: our node's PeerID appeared in both the API response and the crawler output immediately on connection, and was removed from both on disconnect. We repeated these measurements at two points in time (December 2025 and March 2026) and observed no changes in network composition or topology.

\medskip
\noindent\colorbox{gray!12}{\parbox{\dimexpr\linewidth-2\fboxsep\relax}{%
  \textbf{Finding.} The entire Renegade mainnet gossip network consists of \textbf{two clusters}
  and \textbf{four nodes}, all hosted on AS16509 (Amazon.com) in the \texttt{us-east-2}
  availability zone.
}}
\medskip

The barrier to joining is high, with no public bootstrap nodes and sparse documentation, which likely explains why the network has not grown beyond the core team's own infrastructure. The network is therefore not operating as a decentralized dark pool: since relayers hold wallet state in plaintext, this concentration means a single entity effectively observes the trading activity of all of its users.

\subsection{On-Chain Data Collection}
Renegade has active deployments on both Arbitrum One and Base, with DEX volume roughly
comparable across the two chains and alternating in dominance month to month~\cite{defillama_renegade}.
We focus on the \textbf{Base network}, as Arbitrum's faster block times produced over 562 million
blocks over our nine-month analysis window (July 2025--March 2026) compared to $\sim$53 million on
Base, making full execution-trace retrieval substantially more tractable.

Data collection used \textbf{Alchemy}~\cite{alchemy}, a blockchain node and data API provider, in a two-stage pipeline. First, we queried all events
associated with the Renegade Darkpool contract address\footnote{\texttt{0xb4a96068577141749cc8859f586fe29016c935db}} to obtain every
unique transaction hash interacting with the protocol. Second, we performed deep log retrieval for
each hash to capture full transaction receipts and internal logs, including emitted events such as
\textit{Merkle Insertion} and \textit{Nullifier Spent}, and granular ERC-20 token transfers.

\subsection{On-Chain Activity}
\label{sec:concentration}
\begin{figure}[t]
    \centering
   \includegraphics[width=.9\linewidth]{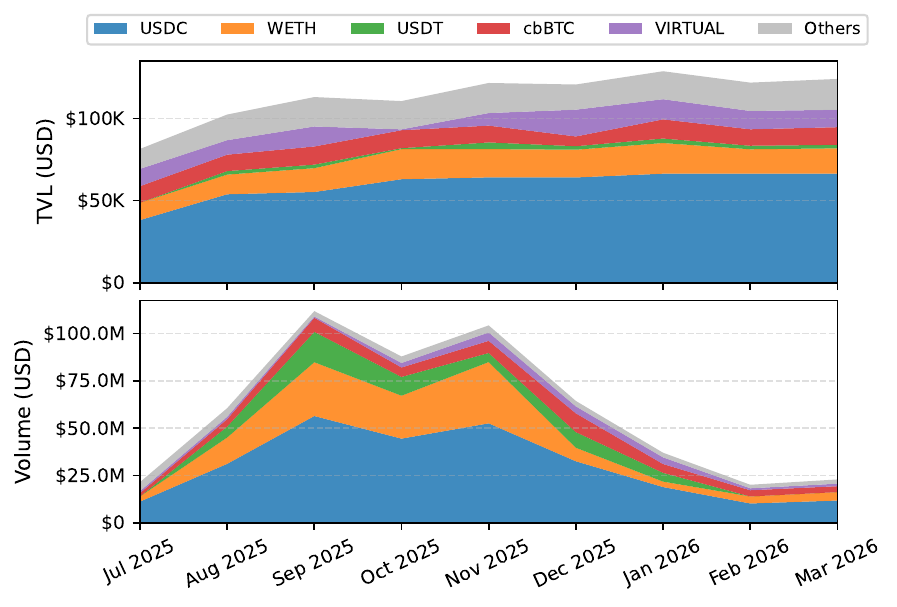}
    \caption{Total Value Locked as all tokens held by the Renegade smart contract (top) and monthly deposit and withdrawal volume (bottom) into the contract, broken down by token and USD-converted daily prices.}
    \label{fig:tvl_volume}
\end{figure}

Figure~\ref{fig:tvl_volume} shows TVL and monthly capital flow over the analysis window, where we use total monthly deposits and withdrawals as a proxy for volume, converted to USD using CoinGecko\footnote{\url{https://www.coingecko.com/}}  daily prices.

TVL in the Renegade contract remains modest, hovering around \$100K USD at any given point. However, the volume of capital moving through the protocol tells a different story: monthly flows peaked at around \$100 million in September and November 2025, and while activity has since declined, flows remain on the order of \$20 million per month. Across the 21 unique tokens observed over the analysis window, volume is highly concentrated: the top 5 tokens (USDC, WETH, USDT, cbBTC, and VIRTUAL) account for 94.6\% of total volume, with USDC alone consistently representing around 50\% in every month.

\begin{figure}[!t]
    \centering
    \includegraphics[width=\linewidth]{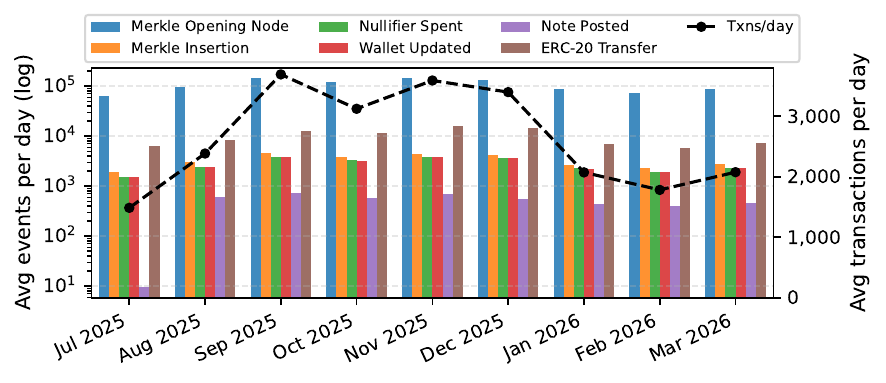}
    \caption{Average daily counts of each event emitted by the Renegade contract on the Base network, and average daily transactions.}
    \label{fig:events_monthly}
\end{figure}

Figure~\ref{fig:events_monthly} shows monthly average daily event and transaction counts, while Table~\ref{tab:events} aggregates these over the full measurement window. The first row counts all transactions interacting with the Renegade contract, and the second is the subset involving direct ERC-20 transfers to or from the contract.\footnote{Transactions may contain logic outside of Renegade (e.g., a router contract that calls Renegade internally). We count only events emitted by the Renegade contract itself, and for token transfers include only amounts moving directly into or out of the Renegade contract address, excluding any further routing.} The remaining rows report protocol events emitted within those transactions.

\begin{table}[!t]
\centering
\footnotesize
\begin{tabular}{lrr}
\hline
Metric & Total & Avg/day \\
\hline
All transactions        & 718{,}627      & $2{,}621 \pm 835$        \\
ERC-20 transfers        & 498{,}040      & $1{,}816 \pm 698$        \\
\hline
Merkle Opening Node     & 28{,}904{,}864 & $105{,}430 \pm 30{,}526$ \\
Merkle Insertion        &    903{,}277   & $3{,}295 \pm 954$        \\
Nullifier Spent         &    766{,}621   & $2{,}796 \pm 879$        \\
Wallet Updated          &    754{,}403   & $2{,}752 \pm 851$        \\
Note Posted             &    135{,}716   & $496 \pm 214$            \\
ERC-20 transfers (events)    &      1{,}013{,}817 & $3{,}696 \pm 1{,}221$ \\
\hline
\end{tabular}
\caption{Totals and average daily rates (mean $\pm$ std across months) for transactions interacting with the Renegade contract and events emitted within them on Base (July 2025--March 2026).
The first two rows count transactions, while remaining rows report protocol events.
}
\label{tab:events}
\end{table}

Merkle Opening Node events dominate at 91.9\% of all events, reflecting wallet membership proofs generated during settlement. Nullifier Spent and Wallet Updated capture wallet state transitions (deposits, withdrawals, order updates, and settlements). Note Posted events, comprising 0.4\% of total activity, provide the closest proxy for completed matches: they are emitted only after successful MPC execution and proof verification, corresponding to approximately $496 \pm 214$ matches per day (roughly one match every 60 blocks at peak, given a 2-second block time).

\subsection{Relay Behavior and Network Concentration}
\label{sec:relay_behavior}

Figure~\ref{fig:senders} shows monthly transaction volume broken down by sender address. Despite 14{,}419 unique senders overall, activity is dominated by three identifiable clusters: Renegade's own relay infrastructure, CoW Protocol solvers~\cite{cow_protocol}, and SwapNet routers~\cite{swapnet}. To link these relay addresses to the user wallets behind them, we trace ERC-20 token transfers within each deposit and withdrawal transaction in December 2025, identifying the on-chain caller and the wallet sending or receiving funds (Figure~\ref{fig:sankey}). In no case was the wallet transferring tokens also the address submitting the transaction, confirming that users never interact with the contract directly. Both the CoW Protocol and SwapNet integrations add an additional observation layer: tokens routed through either aggregator pass through their managed addresses before reaching the user, making the settlement path visible to those operators in addition to the Renegade relay.

\medskip
\noindent\colorbox{gray!12}{\parbox{\dimexpr\linewidth-2\fboxsep\relax}{%
  \textbf{Finding.} Despite 14{,}419 unique senders, the top 15 account for 70.7\% of all transactions, with the single most active address responsible for 12.1\% (60{,}428 transactions). In December 2025, the Dark Pool received $63{,}535$ deposit transactions worth approximately \$31.7M; every deposit was routed through an intermediary. Three clusters dominate: \emph{Renegade Relayers} (five addresses) account for 22\% of deposit transactions ($14{,}045$) but 47\% of deposited volume (\$14.9M), serving $26$ distinct user wallets at an average of $540$ transactions per wallet. \emph{SwapNet} routers (two addresses) contribute 9\% of transactions and 18\% of volume (\$5.8M). \emph{CoW Protocol} solvers (three addresses) are the most active by count---39\% of transactions ($24{,}547$)---but only 15\% of volume (\$4.8M). Together these ten callers cover 69\% of deposit transactions and 80\% of volume.
}}

\medskip
Taken together, the on-chain measurements confirm the centralization observed at the P2P layer: the small set of relay addresses dominating sender activity maps directly to the handful of nodes sustaining the gossip network. The capital flow data makes the relay-to-user mapping visible on-chain, linking wallet addresses to the relay that submits transactions on their behalf. An adversarial counterparty relay can combine this public record with the attacks in Section~\ref{sec:attacks} to build a more complete profile of individual users: on-chain linkage identifies which relay a target uses, while the MPC-layer attacks expose details of their orders and trading history. The volume and token-composition data bound the economic scope of this exposure: roughly 500 matches per day concentrated in five tokens characterizes the order flow at risk in a deployment where independent operators and cross-cluster MPC are the dominant matching mode.

\begin{figure}[t]
    \centering
    \includegraphics[width=\linewidth]{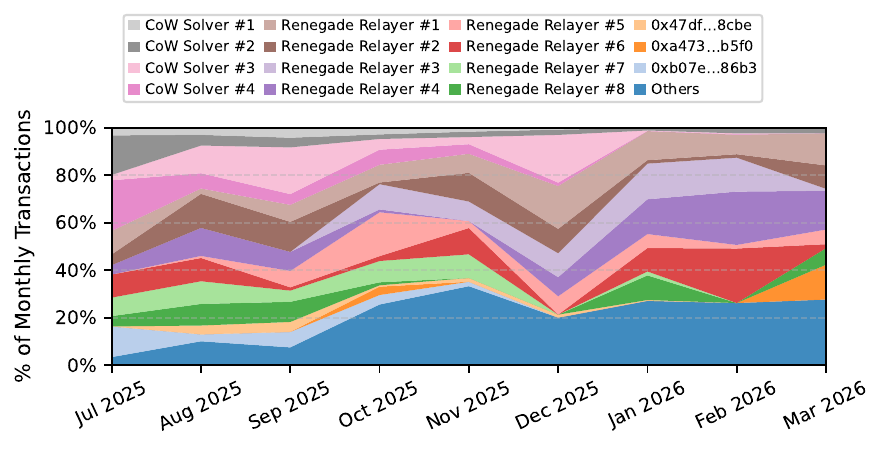}
    \caption{Monthly transaction volume by sender address, July 2025--March 2026, with dominant senders labeled by cluster. }
    \label{fig:senders}
\end{figure}

\begin{figure*}[!t]
    \centering
    \includegraphics[width=0.9\linewidth]{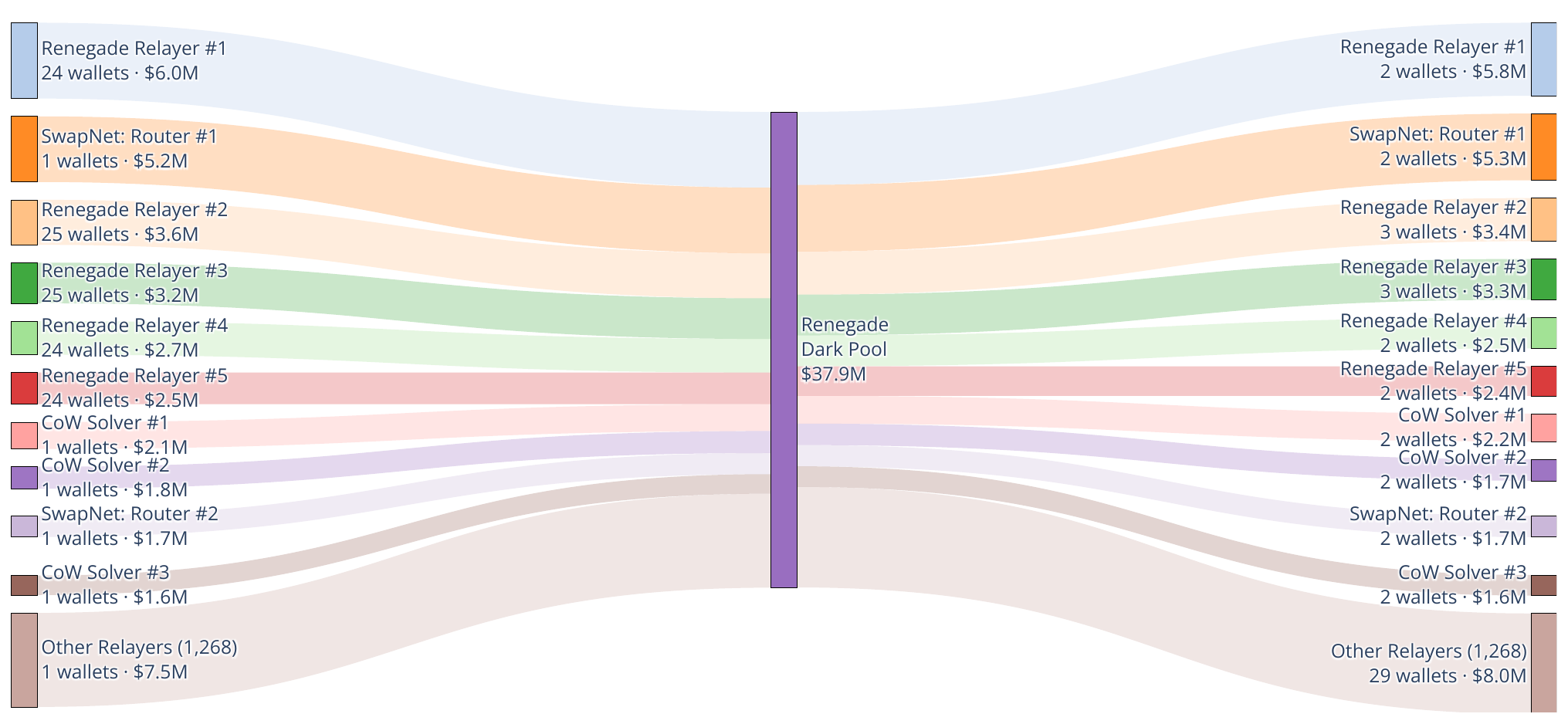}
        \caption{Capital flow through the Renegade Dark Pool contract in December 2025,
    weighted by USD volume of ERC-20 transfers.
    Left nodes show the transaction callers that submitted deposit transactions;
    right nodes show the callers that submitted withdrawal transactions;
    the centre node is the Renegade Dark Pool contract.
    Callers are grouped into three identified clusters---Renegade's own relay
    infrastructure (\emph{Renegade Relayer \#1--5}), CoW Protocol solvers
    (\emph{CoW Solver \#1--3}), and SwapNet routers
    (\emph{SwapNet: Router \#1--2})---plus a residual \emph{Other Relayers} bucket.}
    \label{fig:sankey}
\end{figure*}

\section{Benchmarks}
\label{sec:benchmarks}
The attacks in Section~\ref{sec:attacks} follow from protocol structure alone; whether they are practically efficient is an empirical question. We benchmark the three protocol components governing the cost of Attacks~0, 1, and 3: (i) handshake proof generation, (ii) MPC preprocessing, and (iii) ZKP circuit evaluation; for each, we report honest-execution cost and the asymmetry imposed on the honest party under adversarial behavior.\footnote{Attack~2 requires only passive on-chain observation; its feasibility is established by the empirical analysis in Section~\ref{sec:empirical}.} We find that all four attacks are cheap to execute. The measurements reveal a structural cost asymmetry running through Renegade's design: in each case, the adversary's cost is far lower those cost imposed on the honest party.

\noindent\textbf{Benchmark environment.}
We implement and benchmark Renegade on an Amazon EC2 \texttt{t2.xlarge} instance equipped with 4 vCPUs (Intel Xeon E5-2686 v4 at 2.3\,GHz), 16\,GB RAM, and up to 4\,Gbps network bandwidth, running Ubuntu 24.04. All experiments are implemented in Rust using the \texttt{ark-mpc} framework for SPDZ-style MPC and PLONK-based zkSNARK circuits. MPC experiments use an in-process mock network; reported times reflect computation and serialization overhead only, excluding network latency \footnote{This gives a conservative lower bound on absolute attack cost: latency is an additive term affecting attacker and victim alike, so it does not close the measured cost asymmetries. If anything, latency may favor the adversary in \AtkThree{}, who can hold a session open and delay the invalid input up to the counterparty's response timeout before triggering the abort.}. Wall-clock times are averaged over 10 runs.

\noindent\textbf{Handshake: \emph{how cheap is probing.}}
Both \AtkZero{} and \AtkOne{} originate at the handshake layer, making the handshake RTT the attacker's cost floor. Table~\ref{tab:handshake} shows that proof generation dominates: the \emph{Intent \& Balance} variant gives a computational RTT of $1{,}601 \pm 32$\,ms; \emph{Intent Only} costs $888 \pm 47$\,ms, with verification negligible at $2.1 \pm 0.5$\,ms.\footnote{Each payload is a fixed 966\,B: a 769\,B PLONK proof, a 32\,B wallet-match nullifier $N^{\mathsf{wallet\text{-}match}}$, a 33\,B compressed settlement public key $\mathsf{pk}^{\mathsf{settle}}$, and 132\,B of hiding commitments.}

For \AtkZero{}, this means an adversary can issue $\sim$2 queries per second per relayer with no capital requirement. Given the small active token-pair universe (Section~\ref{sec:empirical}), a full market scan per relayer completes in minutes; applied to all four mainnet nodes simultaneously, it produces a real-time network orderbook map. At larger network scale, the per-probe cost is unchanged and a full scan grows only linearly with the number of active nodes. For \AtkOne{}, the adversary needs only to complete the handshake and the $6.6$\,ms online MPC phase (a total of ${\approx}1.6$\,s) to learn the full match tuple $M$ before deciding whether to abort.

\begin{table*}[t]
\small
\centering
\begin{tabular}{llrr}
    \hline
    Variant & Computation Phase & Wall time (ms) & Payload (B) \\
    \hline
    Intent \& Balance & ProofGen: $\Pi_{\mathrm{commit}}^{R_1}$ & $790 \pm 10$ & 966 \\
    & ProofVrfy: $\Pi_{\mathrm{commit}}^{R_1}$ + ProofGen: $\Pi_{\mathrm{commit}}^{R_2}$ & $809 \pm 20$ & 966 \\
    & ProofVrfy: $\Pi_{\mathrm{commit}}^{R_2}$ & $2.1 \pm 0.5$ & --- \\
    & \textbf{RTT (computational)} & $\mathbf{1{,}601 \pm 32}$ & --- \\
    \hline
    Intent Only & ProofGen: $\Pi_{\mathrm{commit}}^{R_1}$ & $443 \pm 45$ & 966 \\
    & ProofVrfy: $\Pi_{\mathrm{commit}}^{R_1}$ + ProofGen: $\Pi_{\mathrm{commit}}^{R_2}$  & $443 \pm 12$ & 966 \\
    & ProofVrfy: $\Pi_{\mathrm{commit}}^{R_2}$ & $2.1 \pm 0.5$ & --- \\
    & \textbf{RTT (computational)} & $\mathbf{888 \pm 47}$ & --- \\
    \hline 
\end{tabular}
\caption{Renegade handshake performance for a two-round protocol.
Reported times are mean $\pm$ standard deviation (wall-clock, 10 runs).
\textsf{ProofGen} ($\Pi_{\mathrm{commit}}^{R_i}$) denotes the time for relayer $R_i$ to generate and transmit its proof;
\textsf{ProofVerify} ($\Pi_{\mathrm{commit}}^{R_i}$) denotes the time to verify the proof generated by $R_i$ (performed by the counterparty);
RTT denotes the total computational latency across all phases, excluding network delay.}
\label{tab:handshake}
\end{table*}

\begin{table*}[!t]
\small
\centering
\begin{tabular}{llrrrr}
    \hline
    Phase & Sub-phase & Computation & Communication & Rounds & Mult.\ gates \\
    \hline
    Pre-processing (Offline) & Key setup                         & 244\,s  & 5.04\,GB  &  2 & n/a \\
     & Triple gen.\ ($2^{21}$ triples)   & 535\,s  & 15.0\,GB  & 20 & n/a \\
    \hline
    Computation (Online)     & Match circuit$^\dagger$           & $6.6 \pm 0.7$\,ms & 159.1\,KB & 323 & 323 \\
     & Authenticated opening             & ---               &   1.4\,KB &   5 & n/a \\
    \hline
\end{tabular}
\caption{Performance of a single Renegade MPC session (two-party, SPDZ). Preprocessing is non-reusable: an abort due to griefing forces the honest relayer to regenerate a full batch of Beaver triples before the next session. Bandwidth and rounds are reported per party. $^\dagger$ Core match arithmetic only (minimum, fixed-point multiplication, and balance settlement); excludes commitment-consistency checks ($Z_2$) and order validation from the full MPC (Appendix~\ref{app:mpc}). Multiplication gates correspond to Beaver triple consumption per session.}
\label{tab:mpc-session-cost}
\end{table*}

\noindent\textbf{MPC preprocessing: \emph{where griefing leverage arises.}}
\AtkThree{} relies on a structural property of SPDZ: preprocessing is expensive to generate and non-reusable. Table~\ref{tab:mpc-session-cost} shows that the offline phase dominates entirely: key setup ($244 \pm 44$\,s, 5.04\,GB) plus triple generation ($535 \pm 18$\,s, 15.0\,GB) totals 779\,s and 20\,GB per party for a batch of $2^{21}$ Beaver triples. The online match circuit costs only $6.6 \pm 0.7$\,ms and consumes 323 triples per session, so one batch covers ${\approx}6{,}500$ sessions.

This gap is exactly what the griefing attack exploits. An adversary completing a valid handshake then supplying inconsistent MPC inputs forces the target wallet into a locked state and invalidates its entire current triple batch. The attacker spends ${\approx}1.6$\,s; the victim must regenerate 779\,s of preprocessing and loses capacity for ${\approx}6{,}500$ future sessions per wallet. The cost asymmetry is ${\approx}487{\times}$. With only four active mainnet nodes and a small number of active wallets per relay (Section~\ref{sec:empirical}), an adversary can exhaust matching capacity across the entire network with a modest number of concurrent sessions. The $487{\times}$ asymmetry holds regardless of network scale: the attacker can force any wallet into a regeneration cycle at a cost of ${\approx}1.6$\,s, while the victim requires 779\,s to recover before it can match again.

\noindent\textbf{ZKP circuits: the prover/verifier gap.}
Both parties generate a handshake proof, but the adversary controls which variant it submits: an adversary probing with \emph{Intent Only} spends ${\approx}443$\,ms per handshake, while the honest relayer responding with \emph{Intent \& Balance} spends substantially more. Table~\ref{tab:zkp-circuits} reports PLONK proving and verification times across all protocol phases. Handshake circuits dominate: \emph{Intent \& Balance} takes $2{,}888 \pm 660$\,ms ($12{,}345$ constraints); \emph{Intent Only} takes $2{,}226 \pm 830$\,ms. Settlement circuits are lighter ($155$--$763$\,ms). Verification is uniformly fast ($2$--$12$\,ms) across all circuits. Each probe therefore forces the honest relayer to spend $2$--$3$\,s generating $\Pi_{\mathrm{commit}}$ while the adversary pays only ${\approx}443$\,ms, meaning sustained probing can saturate the honest relayer's CPU well below the adversary's query rate.

\medskip
\noindent\colorbox{gray!12}{\parbox{\dimexpr\linewidth-2\fboxsep\relax}{%
\textbf{Takeaway.} The attacks exploit strong cost asymmetries in Renegade's design.
\emph{Probing} saturates honest CPU due to expensive proof generation (2--3\,s) versus negligible verification (2--12\,ms).
\emph{Free-option attacks} reveal full match economics at a cost of only $\approx$1.6\,s per interaction.
\emph{Griefing} forces an honest relayer to discard an entire MPC preprocessing batch (20\,GB, $\approx$6{,}500 sessions) for the same $\approx$1.6\,s attacker cost, yielding a $\approx$487$\times$ asymmetry.
Given the small number of active wallets on the current network, an adversary can continuously force preprocessing regeneration across all active wallets at a cost far below the aggregate recovery time.
}}

\medskip
Next, Section~\ref{sec:mpc_dark_pools} compares Renegade with prior MPC-based dark pools, distinguishing architectural limitations from Renegade-specific design choices.

\medskip

\section{Renegade in the Context of Prior MPC-Based Dark Pools}
\label{sec:mpc_dark_pools}
\textcolor{black}{Renegade is, to our knowledge, the only production MPC-based dark pool that combines permissionless peer-to-peer order discovery with non-custodial on-chain settlement. Prior MPC-based protocols study private matching in more restricted settings (trusted or private order intake, permissioned or server-mediated MPC committees, or institutional intermediaries) and remain proof-of-concept rather than live deployments. On-chain settlement is not itself unique to Renegade: P2DEX~\cite{baum2021p2dex} combines permissionless private MPC-based matching with on-chain settlement, and Rialto~\cite{govindarajan2022privacy} settles on-chain through a permissioned blockchain with broker-mediated execution. What distinguishes Renegade is the combination of permissionless P2P discovery and production deployment: it is the only system in which counterparty discovery, MPC-based matching, and on-chain settlement all operate together under open, adversarial participation.}

Table~\ref{tab:design-comparison}  summarizes representative MPC-based dark pools. Server-based protocols~\cite{cartlidge2019mpc,cartlidge2021multi,da2022kicking,da2022all,chiang2023correlated} rely on trusted or permissioned MPC committees, while institutional deployments such as Prime Match~\cite{polychroniadou2023prime} rely on trusted intermediaries. P2DEX~\cite{baum2021p2dex} uses private order intake with permissionless deployment and on-chain settlement, whereas Rialto~\cite{govindarajan2022privacy} uses broker-mediated MPC over a permissioned blockchain. In contrast, Renegade combines permissionless P2P order discovery, decentralized MPC execution, and non-custodial on-chain settlement in a production system.

\begin{table}[!t]
\centering
\footnotesize
\setlength{\tabcolsep}{4pt}
\renewcommand{\arraystretch}{1.2}
\begin{tabularx}{\columnwidth}{@{}
  >{\hsize=1.4\hsize\raggedright\arraybackslash}X %
  l                                               %
  l                                               %
  >{\hsize=0.6\hsize\raggedright\arraybackslash}X %
  c                                               %
  l                                               %
  @{}}
\hline
\textbf{System} & \textbf{Setting} & \textbf{Discovery} &
\textbf{Matching} & \textbf{Settle} & \textbf{Status} \\
\hline
\rowcolor{yellow!20}
\textbf{Renegade}~\cite{renegade}
& Pless & P2P & 2-PC$^{\dagger}$ & on & Deployed \\
All-for-One~\cite{da2022all}
& Perm & Centr. & $n$-PC & off & PoC \\
MPC-Dark-Side~\cite{cartlidge2019mpc}
& Perm & Centr. & 2--3 srv & off & PoC \\
Correlated-DP~\cite{chiang2023correlated}
& Perm & Curator & MPC$+$DP & off & PoC \\
Turquoise Plato~\cite{cartlidge2021multi}
& Inst & None & Mul-eng & off & PoC \\
Prime Match~\cite{polychroniadou2023prime}
& Inst & Bank & 2-PC & off & Deployed \\
Kick-the-Bucket~\cite{da2022kicking}
& Perm & Intake & 3 srv & off & PoC \\
Rialto~\cite{govindarajan2022privacy}
& Perm & Brokers & $n$-PC & on & PoC \\
P2DEX~\cite{baum2021p2dex}
& Pless
& Intake
& $n$-PC
& on
& PoC \\
\hline
\end{tabularx}

\caption{Architectural comparison of MPC-based dark pools.
Each row represents a system; columns show its permission \emph{Setting}, \emph{Discovery} mechanism, \emph{Matching} configuration, \emph{Settlement} layer, and deployment \emph{Status}.
\emph{Setting:} Pless = permissionless, Perm = permissioned, Inst = institutional.
\emph{Discovery:} P2P = peer-to-peer, Centr. = centralized, Curator = trusted curator, None = no private discovery, Bank = bank-mediated, Intake = private intake.
\emph{Matching:} $n$-PC = $n$-party MPC, srv = servers, DP = differential privacy, $\dagger$ = malicious security.
\emph{Settlement:} on = on-chain, off = off-chain.
\emph{Status:} Deployed = in production, PoC = proof-of-concept.}
\vspace{-2em}
\label{tab:design-comparison}
\end{table}

These architectural differences determine the scope of our findings. 
\AtkZero{} (pre-MPC information leakage) depends on Renegade's permissionless discovery protocol and therefore does not arise in systems with trusted order submission. \AtkTwo{} (on-chain/off-chain correlation) is specific to protocols exposing public on-chain settlement state. In contrast, \AtkOne{} (output fairness) follows from abort-based MPC without fair output delivery, while \AtkThree{} (invalid-input griefing) demonstrates the need for authenticated input validation before MPC execution. These latter issues are not unique to Renegade and should be considered by any MPC-based trading protocol employing similar assumptions. Thus, several of the security properties and resulting vulnerabilities identified in this paper extend beyond a single implementation.

Next, we discuss possible defenses that mitigate these limitations while preserving Renegade's privacy guarantees.

\section{Discussion and Mitigations}
\label{sec:mitigations}
The attacks in Section~\ref{sec:attacks} arise from three structural gaps: pre-MPC leakage during discovery, output unfairness at the MPC conclusion, and execution asymmetries enabling resource exhaustion. Closing all three while preserving low latency, trust minimization, and strong privacy remains an open problem; we therefore address each gap individually, with defenses that raise the cost or reduce the effectiveness of the corresponding attack class.

\noindent\textbf{Privacy Protection (\AtkZero{} and \AtkTwo{}).}
The ideal defense---running all communication at the MPC level, with no handshake---is impractical: MPC is too expensive to attempt with every potential counterparty, and any cheaper pre-filter necessarily leaks some order intent. The mitigations below instead raise the cost of exploiting that leakage.
\begin{itemize}
   \item \textbf{Private Discovery (PSI):} A Private Set Intersection protocol forces an adversary to commit to a set of pairs in advance, raising the cost of a full orderbook scan. Singleton queries still probe one pair at a time, but at higher per-query cost.
    
    \item \textbf{Noise and Dummy Orders:} Injecting dummy orders obscures real trading intent, increasing the number of sessions needed to extract a reliable signal. This raises the cost of probing without eliminating it.
    
    \item \textbf{Batching and Temporal Obfuscation:} Aggregating wallet updates into a single on-chain transaction and adding timing delays breaks the one-to-one mapping between on-chain events and trades, blunting the correlation attack.%

\end{itemize}

\noindent\textbf{Cryptographic Fairness (\AtkOne{}).}
The ideal defense against selective abort is simultaneous output delivery, so neither party can abort after observing the outcome. Achieving this without a trusted third party requires a commitment scheme with guaranteed delivery or a fair exchange protocol, both hard to realize in Renegade's trust model. The mitigations below reduce the informational advantage gained by aborting.

\begin{itemize}
    \item \textbf{Gradual Release:} Replacing the single opening with a multi-round revelation protocol~\cite{even1985randomized} delivers the output incrementally, so no single abort reveals the full match result---at the cost of added round complexity and latency.

   \item \textbf{Financial Penalties (Slashing):} An on-chain deposit before MPC slashed on abort~\cite{kumaresan2016improvements} raises the economic cost of exercising the free option, but effective slashing needs reliable abort attribution, which MPC-with-abort lacks.
    
    \item \textbf{Trusted Delivery:}  Routing output opening through a TEE or a threshold of relay nodes removes any party's unilateral control over delivery, but reintroduces a hardware or multi-party trust assumption.
\end{itemize}

\noindent\textbf{Operational Defenses (\AtkThree{}).}
The ideal defense makes malicious aborts costly and attributable, so an adversary cannot repeatedly force preprocessing regeneration for free. The obstacle is the attribution problem mentioned above: under MPC-with-abort, a deliverate abort is indistinguishable from a network failure. The defenses below instead raise the marginal cost of sustained griefing.

\begin{itemize}
    \item \textbf{Pre-MPC Verification:} A lightweight commitment-consistency check before the full MPC rejects provably inconsistent inputs, filtering the cheapest attack variants; inconsistencies that pass the check can still abort, so this reduces rather than eliminates griefing.

    \item \textbf{Stake and Rate Limiting:} A small deposit to initiate a session plus per-peer rate limits raise the capital and operational cost of sustained attacks. Absent abort attribution the deposit burdens all parties, and rate limits are evadable by rotating identities---but each still raises the attacker's marginal cost per session.

    \item \textbf{Concurrent Match Attempts:} Honest users can participate in multiple matches simultaneously so one griefed session does not block a wallet. This increases computational overhead and may degrade throughput, but ensures eventual execution.
\end{itemize}

The privacy- and fairness-oriented defenses---such as gradual release, batching, and time-locked deposits---improve these properties only by introducing execution delays, at a direct monetary cost in greater price exposure and missed trades.

\section{Related Work}
\label{sec:related_work}
This work relates to three lines of research.

\noindent\textbf{Private Exchange Architectures.}
Privacy-preserving DEXes rely on different primitives with distinct trade-offs. TEE-based systems (e.g., Tristero~\cite{tristero}, Silhouette~\cite{silhouette}) achieve low latency but require hardware trust. ZKP-based designs (e.g., Penumbra~\cite{penumbra}) provide strong state privacy but require external coordination for matching. FHE-based approaches~\cite{singularity} remain impractical for low-latency execution. MPC-based designs enable direct off-chain matching via secret sharing but do not inherently address adversarial economic behavior, leaving strong cryptographic guarantees orthogonal to exploitability in deployed systems.

\noindent\textbf{Fairness and Abort in MPC.}
Fairness is impossible without an honest majority~\cite{cleve1986limits}. In two-party MPC with abort, an adversary may learn the output and terminate before the honest party receives it~\cite{cartlidge2019mpc, da2022kicking}. Financial penalties~\cite{kumaresan2016improvements} and gradual release~\cite{andrychowicz2014fair} mitigate this, but are not enforced in deployed systems, leaving selective-abort strategies economically viable in practice.

\noindent\textbf{Information Leakage and Griefing.}
While MEV on transparent ledgers is well-studied~\cite{daian2020flash, heimbach2023ethereum}, dark pools remain vulnerable to metadata side channels: P2P handshake patterns and on-chain state updates can enable reconstruction of private order flow even when cryptographic protocols are followed correctly. Griefing attacks imposing disproportionate cost are well known in blockchain settings~\cite{buterin2018discouragement}; in MPC dark pools, invalid inputs can force aborts and waste non-reusable preprocessing, creating asymmetries that enable efficient denial of service.

Our work bridges these lines: we formally analyze a deployed MPC dark pool and show that structural vulnerabilities in each category are concretely exploitable against the live network, with cost asymmetries measured empirically.

\section{Ethical Considerations}

This paper involves empirical measurement of a system that was deployed on mainnet at the time of our measurements and analysis of protocol-level vulnerabilities. We address the relevant ethical considerations below.

\medskip
\noindent\textit{Data collection.}
All on-chain data analyzed in this work is publicly available on the Base blockchain. No private or privileged data sources were used. While our analysis links relay addresses to wallet addresses, these associations are directly derivable from plain transaction metadata and do not require any information beyond what is visible to any blockchain observer.

\medskip
\noindent\textit{P2P network measurement.}
Our P2P measurements were conducted by joining the public Renegade gossip network as a standard participant. The \texttt{/v0/network} API endpoint queried is publicly accessible without authentication, and crawling was performed using the Nebula libp2p crawler in read-only mode. Neither activity modifies network state or disrupts service to other participants.

\medskip
\noindent\textit{Attack implementation.}
The attacks described in Section~\ref{sec:attacks} are formal protocol-level analyses and were not implemented or executed. Section~\ref{sec:benchmarks} benchmarks the underlying protocol components (handshake proof generation, MPC session cost, and ZKP circuit evaluation) on our own infrastructure to establish cost parameters; no attack was directed at the live Renegade network or at any real users.

\medskip
\noindent\textit{Responsible disclosure.}
On August 12, 2026, before making this work public, we notified the Renegade development team of the vulnerabilities identified in this work and shared a draft of the full manuscript containing the technical details of our analysis, giving them time to review the findings. We did not execute the attacks against the live network or any real users. The developers acknowledged our report, raised no objection to publication, and confirmed that some of the concerns we describe were known to them. In particular, they indicated that selective abort could be addressed through mechanisms such as peer scoring or exponential backoff combined with time-lock encryption, in line with the mitigations we discuss in Section~\ref{sec:mitigations}. They further confirmed that, in the deployment, the relayer network remained restricted to nodes operated by the team, consistent with our network measurements (Section~\ref{sec:empirical}). Our attacks therefore target the protocol as specified, in its intended permissionless setting, rather than a particular operational configuration.

\bibliographystyle{IEEEtran}
\bibliography{references}

\appendices

\section{Open Science}
We provide all artifacts necessary to reproduce the results presented in this paper in the following public repository:
\href{https://zenodo.org/records/22810597}
{https://zenodo.org/records/22810597}
The artifact package includes:

\medskip
\noindent\textit{On-chain dataset.} The complete transaction-level dataset collected from the Renegade Base deployment (Section~\ref{sec:empirical}), including all transaction hashes, protocol events, and ERC-20 transfer logs retrieved via Alchemy~\cite{alchemy}.

\medskip
\noindent\textit{Data collection and analysis scripts.} All scripts used to collect on-chain data and reproduce the figures and tables in Section~\ref{sec:empirical}, including the two-stage Alchemy pipeline for event and transaction retrieval.

\medskip
\noindent\textit{Benchmark code.} All code used to measure handshake proof generation, MPC session cost, and ZKP circuit evaluation (Section~\ref{sec:benchmarks}), implemented in Rust using the \texttt{ark-mpc} framework and PLONK-based zkSNARK circuits.

\medskip
\noindent\textit{P2P network measurement.} Current Renegade mainnet topology is publicly accessible via the API endpoint documented in Section~\ref{sec:empirical}, and the public Nebula libp2p crawler~\cite{nebula}.

\section{Generative AI Disclosure}

This work used large language model (LLM) assistants in a supporting role. Claude (Anthropic) was used for prose rewriting, grammar correction, and cross-section consistency, and for assistance with data analysis workflows and plotting scripts. ChatGPT (OpenAI) was used for occasional small-scale text editing. In all cases, the authors verified and took responsibility for all outputs. All technical content, analyses, formal arguments, and conclusions are original work by the authors; no LLM was used to generate claims, proofs, or experimental results.

\section{Details on Renegade Protocol}
\label{appendix:trade_lifecycle}
\subsection{Trade Lifecycle: Full Description}
The trade lifecycle in Renegade \cite{renegade} consists of two distinct operational domains: \emph{on-chain state transitions}, where users interact with the smart contract to manage wallet state, and \emph{off-chain matching}, where relayers coordinate to discover and execute trades using secure computation. Figure~\ref{fig:trade_cycle_part1} and ~\ref{fig:tradecycle_part2} illustrate the complete process from wallet initialization to settlement. 

A trade in Renegade proceeds through the following stages:
(i) wallet creation,
(ii) wallet update,
(iii) delegation to a relayer,
(iv) counterparty discovery,
v) handshake,
(vi) secure matching via MPC and proof generation, vii) on-chain submission, and
(viii) trade settlement.

The detailed workflow is as follows:

\medskip
\noindent\emph{\textbf{Step 1: Wallet Creation.}}
The wallet state is defined as 
\[W := (B, O, F, K, r)
\]
(Section~\ref{sec:trade_life_cycle}). The user initializes $W$ with zero balances, orders, and fees (i.e., $B, O, F = 0$), a set of access keys $K$, and randomness $r$. The user computes a commitment $C(W)$ and generates a proof $\Pi_{\mathsf{create}}$ that $C(W)$ corresponds to a correctly initialized wallet (see~\cite{renegade}). The pair $(C(W), \Pi_{\mathsf{create}})$ is submitted to the smart contract, which verifies the proof and inserts $C(W)$ into the global commitment Merkle tree.
The trade flow begins at \Circled{1}  in Figure~\ref{fig:trade_cycle_part1}).

\medskip
\noindent\emph{\textbf{Step 2: Wallet Update.}}
To modify the wallet state (e.g., deposits, withdrawals, or order updates), the user transitions from $W$ to a new state $W'$ with commitment $C(W')$. This update invalidates the previous state by revealing the nullifiers
\[
\{N^{\textsf{wallet-spend}}(W),\; N^{\textsf{wallet-match}}(W)\}.
\]

The user generates a proof $\Pi_{\mathsf{update}}$ that certifies: (i) $C(W)$ is a valid commitment in the global tree, (ii) both $C(W)$ and $C(W')$ are well-formed, (iii) the nullifiers are correctly derived and unused, and (iv) the transition from $W$ to $W'$ is valid. The user submits $(C(W'), \Pi_{\mathsf{update}}, \mathrm{pk}^{\textsf{settle}})$ to the smart contract. If verification succeeds and the nullifiers are unused, the contract appends $C(W')$ to the commitment tree and marks the corresponding nullifiers spent (see \Circled{2} in Figure~\ref{fig:trade_cycle_part1}.)

\medskip
\noindent\emph{\textbf{Step 3: Delegation to Relayer.}}
After the updated wallet state is recorded on-chain, the user delegates matching authority to a relayer by securely sharing $(W', C(W'))$ and a matching key $\mathrm{sk}^{\mathsf{match}}$ over a private channel. The relayer verifies consistency between $W'$ and $C(W')$ and maintains a local view of the wallet state. Using the delegated key, it can participate in order matching and assist in settlement, but cannot modify the wallet state as it does not possess the root secret key $\mathrm{sk}^{\mathsf{root}}$ (see \Circled{3} in Figure~\ref{fig:trade_cycle_part1}).

\begin{figure*}[!b]
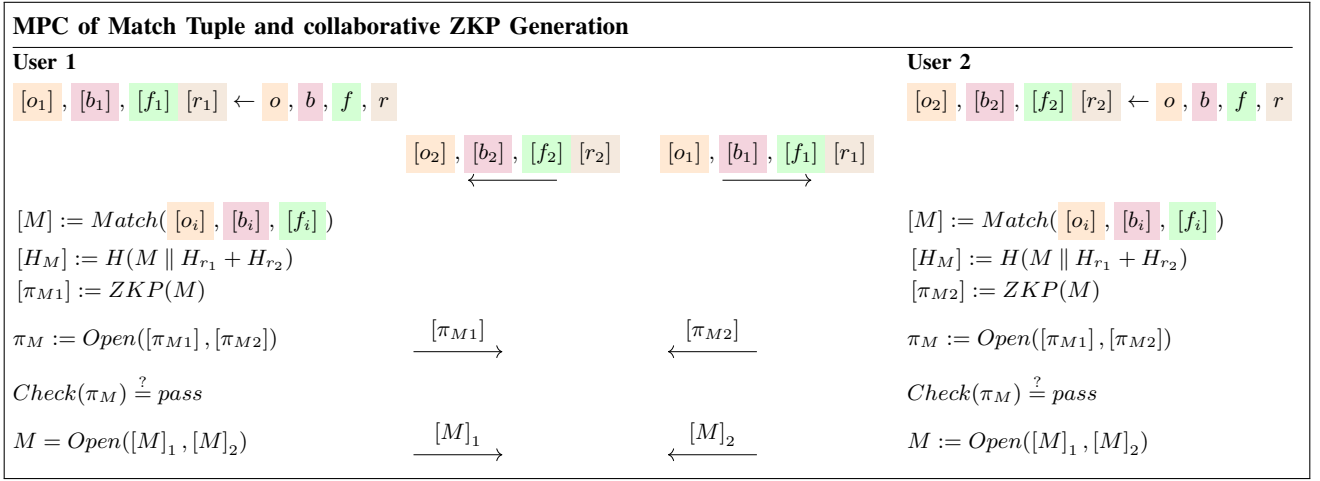

    \centering
    \begin{pchstack}[boxed]
        \resizebox{.94\textwidth}{!}{%
            \procedure{\textbf{MPC of Match Tuple and collaborative ZKP Generation}}{
                \textbf{User 1} \< \< \< \textbf{User 2} \\
                \eqnmarkbox[orange]{}{[o_1]}, \eqnmarkbox[purple]{}{[b_1]},  \eqnmarkbox[green]{}{[f_1]} \eqnmarkbox[brown]{}{[r_1]} \gets \eqnmarkbox[orange]{}{o}, \eqnmarkbox[purple]{}{b}, \eqnmarkbox[green]{}{f}, \eqnmarkbox[brown]{}{r} \< \< \< \eqnmarkbox[orange]{}{[o_2]}, \eqnmarkbox[purple]{}{[b_2]},  \eqnmarkbox[green]{}{[f_2]} \eqnmarkbox[brown]{}{[r_2]} \gets \eqnmarkbox[orange]{}{o}, \eqnmarkbox[purple]{}{b}, \eqnmarkbox[green]{}{f}, \eqnmarkbox[brown]{}{r} \\
                \< \sendmessageleft*[4em]{\eqnmarkbox[orange]{}{[o_2]}, \eqnmarkbox[purple]{}{[b_2]},  \eqnmarkbox[green]{}{[f_2]} \eqnmarkbox[brown]{}{[r_2]}} \<  \sendmessageright*[4em]{\eqnmarkbox[orange]{}{[o_1]}, \eqnmarkbox[purple]{}{[b_1]},  \eqnmarkbox[green]{}{[f_1]} \eqnmarkbox[brown]{}{[r_1]}} \< \\
                \left[ M \right]  := Match( \eqnmarkbox[orange]{}{[o_i]}, \eqnmarkbox[purple]{}{[b_i]},  \eqnmarkbox[green]{}{[f_i]} ) \< \< \< [M] := Match( \eqnmarkbox[orange]{}{[o_i]}, \eqnmarkbox[purple]{}{[b_i]},  \eqnmarkbox[green]{}{[f_i]} )\\
                \left[ H_M \right] := H(M \mathbin\| H_{r_1} +  H_{r_2}) \< \< \< \left[ H_M \right] := H(M \mathbin\| H_{r_1} +  H_{r_2})\\
                \left[ \pi_{M1} \right] := ZKP(M) \< \< \< \left[ \pi_{M2} \right] := ZKP(M) \\
                \pi_{M} := Open( \left[ \pi_{M1} \right],  \left[ \pi_{M2} \right] ) \< \sendmessageright*[4em]{\left[ \pi_{M1} \right]} \< \sendmessageleft*[4em]{\left[ \pi_{M2} \right]} \< \pi_{M} := Open( \left[ \pi_{M1} \right],  \left[ \pi_{M2} \right] ) \\
                Check(\pi_{M}) \overset{?}= pass \< \< \< Check(\pi_{M}) \overset{?}= pass  \\
                M = Open( \left[ M \right]_{1}, \left[ M \right]_{2} ) \< \sendmessageright*[4em]{\left[ M \right]_{1}} \< \sendmessageleft*[4em]{\left[ M \right]_2}
                \< M := Open( \left[ M \right]_{1},  \left[ M \right]_{2} )
             }
        }
    \end{pchstack}
    \caption{MPC of Match Tuple and Collaborative Proof generation in Renegade Protocol. The process proceeds as follows: 
        (1) \textbf{Secret Sharing}: Both users decompose their private state ($o_i, b_i, f_i$) into Shamir secret shares $[ \cdot ]$ and exchange them. 
        (2) \textbf{MPC Execution}: Relayers compute secret shares of the match tuple $[M]$ and a hiding commitment $[H_M]$ without revealing the underlying data. 
        (3) \textbf{Collaborative ZKP}: Both parties locally generate shares of a zero-knowledge proof ($\pi_{Mi}$), then exchange and open them to verify the integrity of the MPC execution ($Check(\pi_M)$). 
        (4) \textbf{Opening Phase}: The final stage requires exchanging shares of the match result $[M]_i$.}
    \label{fig:mpc_match_tuple}
\end{figure*}

\medskip
\noindent\emph{\textbf{Step 4: Counterparty Discovery.}}
The relayer gossips its order intent over the P2P network (see Section~\ref{sec:network}). A relayer with a potentially compatible order responds via a direct channel, establishing a candidate counterparty.

\medskip
\noindent\emph{\textbf{Step 5: Handshake with the counterparty.}}
To construct a handshake payload, relayer $R_1$ selects an order $o \in O$, balance $b \in B$, and fee tuple $f$ according to its policy. It samples randomness $r$ and computes hiding commitments $(H_o, H_b, H_f, H_r)$. It then generates a proof $\Pi_{\mathsf{commit}}$ certifying that the commitments are well-formed and correspond to an unspent wallet $W'$ containing $(o, b, f)$ with sufficient balance to cover fees. 
$R_1$ forms the handshake payload
\[
H_1 =
\big(
\Pi_{\mathsf{commit}},
N^{\textsf{wallet-match}}(W'),
pk^{\textsf{settle}},
H_o, H_b, H_f, H_r
\big),
\]
which binds it to the selected inputs while preserving privacy, and sends $H_1$ to a counterparty $R_2$. Upon receipt, $R_2$ verifies the proof, checks that the nullifier $N^{\textsf{wallet-match}}(W')$ is unused, and ensures the handshake has not been previously marked as failed.

If valid, $R_2$ responds with its own handshake payload $H_2$ over a secure channel (see \Circled{5} in Figure~\ref{fig:tradecycle_part2}). After verifying $H_2$, $R_1$ proceeds to MPC-based matching.

\medskip

\noindent\emph{\textbf{Step 6: Secure Matching via MPC and Proof Generation.}}

\noindent\emph{Matching.}
Renegade performs order matching using a maliciously secure SPDZ-style two-party computation between relayers $(R_1, R_2)$ established during the handshake. The protocol evaluates the matching function over secret-shared inputs (orders, balances, and fees) while keeping all intermediate values private, producing a match tuple along with validity indicators (see Appendix~\ref{app:mpc}).

\noindent\emph{Proof generation.}
After executing the circuit, $R_1$ and $R_2$ jointly generate a zero-knowledge proof $\Pi_{\mathsf{match}}$ within the MPC framework. This binds the output $M$ to the public commitments via $H_M$ (see \Circled{6} in Figure~\ref{fig:tradecycle_part2}).

\noindent\emph{Opening phase.}
The relayers then perform an authenticated opening of the MPC outputs by exchanging shares. Integrity is enforced via SPDZ MAC verification, which detects any deviation. Upon successful verification, both parties obtain $M$, $H_M$, consistency flags $(Z_1, Z_2)$, and the proof $\Pi_{\mathsf{match}}$.

\medskip
\noindent\emph{\textbf{Step 7: On-Chain Submission.}}

\noindent\emph{Match note construction and encryption.}
If the output indicates a valid match, each relayer derives a match note from $M$ encoding the trade parameters (assets, amounts, direction, and metadata). The relayers construct notes $N_1$ and $N_2$ and encrypt them as
$E_{pk_{R_1}^{\textsf{settle}}}(N_1)$ and $E_{pk_{R_2}^{\textsf{settle}}}(N_2)$. Each relayer also generates a proof $\Pi_{\mathsf{enc}}$ to ensure correct formation and encryption of the notes, consistent with $M$ and $H_M$ (see \Circled{7a} in Figure~\ref{fig:tradecycle_part2}).

\noindent\emph{On-chain submission.}
$R_1$ submits $\Pi_{\mathsf{commit}}^{R_1}$, $\Pi_{\mathsf{commit}}^{R_2}$, $\Pi_{\mathsf{match}}$, and $\Pi_{\mathsf{enc}}$, along with the corresponding public inputs, to the smart contract. Upon verification, the contract marks the wallet-match nullifiers as spent and inserts the encrypted notes into the commitment tree. This step, referred to as \emph{encumbering}, completes the joint execution phase. Settlement is then performed independently by each party (see \Circled{7b} in Figure~\ref{fig:tradecycle_part2}).

\medskip
\noindent\emph{\textbf{Step 8: Trade Settlement.}}
To settle a trade, a relayer (e.g., $R_1$) retrieves its encrypted note and decrypts it using its settlement key $sk^{\textsf{settle}}$. It then constructs an updated wallet state $W''$, preserving the fee parameters $f$ and keys $K$ from $W'$.
The relayer generates a proof $\Pi_{\mathsf{settle}}$ and submits it to the smart contract. Upon verification, the contract marks the \emph{wallet-spend} and \emph{note-redeem} nullifiers as spent, completing settlement for that relayer.

Either of the parties submit valid settlement proofs, all relevant nullifiers (wallet-match, wallet-spend, and note-redeem) are consumed, preventing further use of the previous wallet state $W'$ or the corresponding notes (see \Circled{8} in Figure~\ref{fig:tradecycle_part2}).

\subsection{MPC-based Matching in Renegade.}
\label{app:mpc}

Renegade performs order matching via maliciously secure SPDZ-based two-party computation followed by collaborative ZKP proof generation before revealing the match tuple, as shown in Figure~\ref{fig:mpc_match_tuple}.
Renegade's MPC flow comprises three phases: (A) preprocessing generates Beaver triples and global MAC key shares; (B) computation evaluates the matching circuit over secret-shared inputs (order, balance, fees, randomness) to produce match tuple shares; and (C) opening reveals the tuple by exchanging shares.
  
\begin{enumerate}
\item \noindent\underline{\textbf{Pre-processing}.}
Generates input-independent Beaver triples $([a],[b],[c])$ with $c = a \cdot b$ and MAC key shares $[\alpha]$ for secure computation.

\medskip
\item \noindent\underline{\textbf{Computation}.}
The functionality evaluates the circuit $\mathcal{C}_{\mathsf{match}}$ on authenticated shares. The circuit takes as public inputs the handshake commitments and auxiliary market data, and computes the match over secret-shared inputs without revealing intermediate values.
\begin{itemize}
\item \noindent\textbf{Public Inputs.}
\begin{itemize}
\item Handshake commitments:
$H_o^{(i)}, H_b^{(i)}, H_f^{(i)}, H_r^{(i)}$ for $i \in \{1,2\}$
\item Market data: midpoint price $P_{\mathsf{mid}}$, protocol parameters
\end{itemize}

\item \noindent \textbf{Private Inputs (secret-shared).}
Each party $R_i$ holds:
\[
x_i = (o_i, b_i, f_i, r_i)
\]
where: $o_i = (m_i^{\mathrm{in}}, m_i^{\mathrm{out}}, v_i, p_i, t_i)$ order, $b_i$ = balance vector, $f_i$ = fee parameters, $r_i$ = randomness used in commitments.

Each secret-shared value $[x]$ is associated with a MAC share under a global secret key $\alpha$, enabling integrity checks during opening. 

Inputs are encoded as additive shares over $\mathbb{F}_p$: $
x_i = [x_i]_1 + [x_i]_2 $
with corresponding MAC shares under global key $\alpha$: $\gamma_i = \alpha \cdot x_i$
\end{itemize}

\medskip
The circuit performs these checks over authenticated shares:
\begin{itemize}
\item \noindent \textbf{Commitment consistency ($Z_2$).}
Recomputes commitments from private inputs and verifies against public handshake values.
computes \[\hat{H}_{\ast}^{(i)} = \mathsf{Com}(x_i, r_i)\] for $x_i \in \{o_i, b_i, f_i\}$ and verifies $\hat{H}_{\ast}^{(i)} = H_{\ast}^{(i)}$ for $i \in \{1,2\}$,
yielding
\[
Z_2 = \land_{i \in \{1,2\}}
\left(\hat{H}_{\ast}^{(i)} = H_{\ast}^{(i)}\right).
\]

\item \noindent \textbf{Order compatibility.}
Ensures $Z_{\mathrm{ord}}= (m_1^{\mathrm{out}}  = m_2^{\mathrm{in}})$ and $(m_1^{\mathrm{in}} = m_2^{\mathrm{out}})$ with price/expiration constraints.

\item \noindent \textbf{Match volume.}
Computes executable volume $\hat{v} = \min(v_1, v_2)$ and direction $\hat{d} \in \{0,1\}$.

\item \textbf{Balance/fees.} Verifies $b_i \geq \hat{v} + \mathrm{fee}(f_i)$ and computes $f_i' = \mathrm{Fee\_fun}(f_i, \hat{v})$.
 
\item \textbf{Validity ($Z_1$).} Sets $Z_1 = Z_{\mathrm{ord}} \wedge (\hat{v} > 0)$.

\item \textbf{Output ($M$).} Constructs 
\[M = (\hat{m}_1, \hat{m}_2, \hat{v}_1, \hat{v}_2, \hat{d}, f_1, f_2)\] with commitment $H_M = \mathsf{Com}(M)$.

\item \textbf{Mask:} Applies $M \leftarrow Z_1 \cdot Z_2 \cdot M$.
\end{itemize}

All operations are performed over authenticated shares using Beaver triples.

\medskip
\item \noindent\underline{\textbf{Opening.}}
Parties exchange shares to reconstruct and MAC-verify the match tuple $(M, H_M, Z_1, Z_2)$.

\begin{itemize}
\item exchanging shares
\item verifying MACs: $\gamma = \alpha \cdot x$
\end{itemize}

If verification fails, output $\perp$. Otherwise, protocol outputs
$(M, H_M, Z_1, Z_2)$.
\end{enumerate}

\medskip
\subsection{ZKP Circuits in Renegade}
Table~\ref{tab:zkp-circuits} reports the proving and verification times for all Plonk circuits used in Renegade.

\newpage

\begin{table*}[htb!]
\centering
\small
\resizebox{0.80\textwidth}{!}{%
\begin{tabular}{llrll}
\hline
\textbf{Phase} & \textbf{Circuit} & \textbf{Gates} & \textbf{Prove Time} & \textbf{Verify Time} \\
\hline
\multirow{4}{*}{\shortstack[l]{Pre-MPC\\(Handshake,\\$\Pi_{\mathrm{commit}}$)}}
  & Intent \& Balance           & 12{,}345 & $2888 \pm 660$\,ms & $3.6 \pm 0.3$\,ms \\
  & Intent \& Balance (1st fill) & 12{,}978 & $3012 \pm 778$\,ms & $4.6 \pm 0.6$\,ms \\
  & Intent Only                 &  5{,}577 & $2226 \pm 830$\,ms & $4.1 \pm 0.3$\,ms \\
  & Intent Only (1st fill)      &  1{,}512 & $ 465 \pm  29$\,ms & $5.6 \pm 0.3$\,ms \\
\hline
\multirow{7}{*}{\shortstack[l]{Post-MPC\\(Settlement)}}
  & Intent \& Balance --- Public settlement   & 1{,}116 & $415 \pm  24$\,ms & $ 6.3 \pm 1.2$\,ms \\
  & Intent \& Balance --- Private settlement  & 2{,}779 & $763 \pm  75$\,ms & $ 5.7 \pm 0.6$\,ms \\
  & Intent \& Balance --- Bounded settlement  &   701   & $247 \pm  26$\,ms & $ 4.6 \pm 0.5$\,ms \\
  & Intent Only --- Public settlement         &   372   & $291 \pm  20$\,ms & $11.8 \pm 4.4$\,ms \\
  & Intent Only --- Bounded settlement        &   322   & $155 \pm  40$\,ms & $ 8.2 \pm 2.1$\,ms \\
  & Output Balance                            & 6{,}771 & $1434 \pm 290$\,ms & $3.4 \pm 0.3$\,ms \\
  & New Output Balance                        & 11{,}883 & $2918 \pm 307$\,ms & $3.4 \pm 0.4$\,ms \\
\hline

\multirow{9}{*}{\shortstack[l]{Wallet\\Management}}
  & Valid Balance Create              & --- & $1553 \pm 504$\,ms & $5.3 \pm 0.2$\,ms \\
  & Valid Deposit                     & --- & $1395 \pm  72$\,ms & $3.6 \pm 0.4$\,ms \\
  & Valid Withdrawal                  & --- & $1414 \pm  41$\,ms & $4.4 \pm 0.4$\,ms \\
  & Valid Order Cancellation          & --- & $1396 \pm  84$\,ms & $4.5 \pm 0.6$\,ms \\
  & Valid Note Redemption             & --- & $ 759 \pm 104$\,ms & $4.6 \pm 0.8$\,ms \\
  & Valid Private Relayer Fee Payment & --- & $1418 \pm 142$\,ms & $4.2 \pm 0.6$\,ms \\
  & Valid Public Relayer Fee Payment  & --- & $1415 \pm  47$\,ms & $3.4 \pm 0.4$\,ms \\
  & Valid Private Protocol Fee Payment & --- & $2915 \pm 436$\,ms & $3.6 \pm 0.3$\,ms \\
  & Valid Public Protocol Fee Payment  & --- & $1416 \pm  77$\,ms & $4.3 \pm 0.4$\,ms \\
\hline
\end{tabular}}
\caption{Renegade ZKP circuits classified by protocol phase. Prover and verifier times represent
the mean $\pm$ standard deviation (wall-clock time, single-prover, 10 samples). Gates
denote PLONK constraint counts; all proofs are 769\,B. Wallet-management circuits run
independently of the MPC matching flow. Here, \emph{Intent \& Balance} validates intent
presence in the Merkle tree with a paired balance account, while \emph{Intent Only} verifies only intent presence in the Merkle tree.}
\label{tab:zkp-circuits}
\end{table*}

\end{document}